\documentclass[lettersize,journal]{IEEEtran}

\usepackage{amsmath,amsfonts}
\usepackage{CJKutf8}
\usepackage{mathtools}                
\usepackage{mathrsfs}                 
\usepackage{stmaryrd}                 
\usepackage{algorithmic}
\usepackage{graphicx}
\usepackage{textcomp}
\usepackage{xcolor}
\usepackage{subfigure}
\usepackage{amsthm}
\usepackage{url}
\usepackage{threeparttable}
\usepackage{framed}
\usepackage{multicol}
\usepackage{mdwlist}
\usepackage{booktabs}
\usepackage{changepage}
\usepackage{array}
\usepackage{tabularx}
\usepackage{threeparttable}
\usepackage{caption}
\usepackage{multirow}
\usepackage{colortbl, hhline}
\usepackage{wasysym} 
\usepackage{pifont}
\newcommand{\RNum}[1]{\uppercase\expandafter{\romannumeral #1\relax}}

\newtheorem{mytheorem}{Theorem}

\newtheorem{definition}{Definition}

\def\BibTeX{{\rm B\kern-.05em{\sc i\kern-.025em b}\kern-.08em
    T\kern-.1667em\lower.7ex\hbox{E}\kern-.125emX}}
\usepackage{balance}

\AtBeginDocument{%
  \providecommand\BibTeX{{%
    Bib\TeX}}}

\ifCLASSINFOpdf
\else
\fi
\begin{document}
%
\title{CCN: Decentralized Cross-Chain Channel Networks Supporting Secure and Privacy-Preserving Multi-Hop  Interactions}

\author{
Minghui Xu,
Yihao Guo, 
Yanqiang Zhang,
Zhiguang Shan,  
Guangyong Shang,
Zhen Ma, \IEEEmembership{Member, IEEE},
Bin Xiao,  \IEEEmembership{Fellow, IEEE},
and Xiuzhen Cheng, \IEEEmembership{Fellow, IEEE}
\thanks{
M. Xu and X. Cheng are with the School of Computer Science and Technology,
Shandong University, Qingdao, Shandong, P.R. China 
(e-mail: \{mhxu, xzcheng\}@sdu.edu.cn).}
\thanks{
Y. Guo and B. Xiao are with the Department of Computing,
The Hong Kong Polytechnic University, Hong Kong SAR, P.R. China 
(e-mail: \{yihaoguo, b.xiao\}@polyu.edu.hk).}
\thanks{
Y. Zhang and Z. Shan are with the Department of Informatization and Industry Development,
State Information Center, Beijing, P.R. China 
(e-mail: \{zhangyq, shanzg\}@sic.gov.cn).}

\thanks{
G. Shang and Z. Ma is with Inspur Yunzhou Industrial Internet Co., Ltd, Jinan 250101, China (email: \{shangguangyong, mazhenrj\}@inspur.com).}

\thanks{Corresponding author: Yihao Guo.}
}


\maketitle

\begin{abstract}
Cross-chain technology enables interoperability among otherwise isolated blockchains, supporting interactions across heterogeneous networks. Similar to how multi-hop communication became fundamental in the evolution of the Internet, the demand for multi-hop cross-chain interactions is gaining increasing attention. This extended interaction model builds upon earlier single-hop designs and requires coordination among multiple intermediate nodes to ensure correct cross-chain execution. However, this growing demand introduces new security and privacy challenges. On the security side, multi-hop interactions depend on the availability of multiple participating nodes. If any node becomes temporarily offline during execution, the protocol may fail to complete correctly, leading to settlement failure or fund loss. On the privacy side, the need for on-chain transparency to validate intermediate states may unintentionally leak linkable information, compromising the unlinkability of user interactions. In this paper, we propose the Cross-Chain Channel Network (CCN), a decentralized network designed to support secure and privacy-preserving multi-hop cross-chain transactions. Through experimental evaluation, we identify two critical types of offline failures, referred to as active and passive offline cases, which have not been adequately addressed by existing solutions. To mitigate these issues, we introduce R-HTLC, a core protocol within CCN. R-HTLC incorporates an hourglass mechanism and a multi-path refund strategy to ensure settlement correctness even when some nodes go offline during execution. Importantly, CCN addresses not only the correctness under offline conditions but also maintains unlinkability in such adversarial settings. To overcome this, CCN leverages zero-knowledge proofs (ZKPs) and off-chain coordination, ensuring that interaction relationships remain indistinguishable even when certain nodes are temporarily offline. Finally, our experiments demonstrate that CCN effectively prevents prolonged fund unavailability and mitigates the risk of fund loss caused by offline failures. 

\end{abstract}



%

\section{Introduction}~\label{sec:Introduction}
%
%
%
Cross-chain technology has become a key enabler for interaction among otherwise isolated blockchains, supporting a wide range of applications across heterogeneous networks~\cite{belchior2021survey,lao2020survey,xu2024exploring,augusto2024sok}. 
As of November 2025, CoinMarketCap records more than ten thousand active cryptocurrencies~\cite{CoinMarketCap}.  
The rise in blockchain deployments has resulted in higher cross-chain demand, as shown by the growth of cross-chain transaction volume from \$18.6~billion in September 2024 to nearly \$50~billion in November 2024, an increase of about 168\%. Similar to how multi-hop routing became fundamental in the development of the Internet, multi-hop cross-chain interaction is emerging as a practical and necessary approach for achieving end-to-end interoperability. Multiple projects, including Cosmos, Polkadot, and Uniswap, have recognized this need in their design goals or technical roadmaps~\cite{IBC,XCM,uniswap}, highlighting the growing importance of supporting multi-hop coordination across chains. 

To meet the growing demand for large-scale cross-chain interaction, a recent line of work has introduced cross-chain channels~\cite{guo2023cross,jia2023cross,zhang2023cross,luo2024crosschannel,guo2025xrwa} to resolve the trade-off between centralization and scalability.  
This design can be integrated into existing cross-chain frameworks, yielding about a factor of $N$ efficiency improvement when processing $N$ transactions~\cite{guo2023cross}. 
However, similar to most existing solutions~\cite{zamyatin2019xclaim, tsabary2021mad, thyagarajan2022universal, xie2022zkbridge, hanzlik2024sweep, guo2024zkCross, han2024p2c2t}, current cross-chain channels focus primarily on interactions between two chains. When extending to the multi-hop domain, new security and privacy challenges emerge, which are not adequately addressed. 
This motivates us to explore whether it is possible to design a multi-hop cross-chain channel network that inherits the efficiency of existing cross-chain channel designs while, more importantly, satisfying the security and privacy requirements of multi-hop interaction. 

\subsection{Where Existing Approaches Fall Short}
A cross-chain channel consists of four phases.  
First, both parties lock their assets to open the channel.  
Second, the parties exchange and update the state off-chain.  
Third, the channel is settled on-chain and then closed.  
If settlement fails, the protocol enters a timeout phase in which both parties refund their locked assets. A multi-hop cross-chain channel involves more than two parties in every phase of the protocol, introducing the following additional challenges (\textbf{[C1-C3]}). 

\noindent \textbf{[C1] Active offline issues.} %
Multi-hop interactions rely on intermediary nodes to relay messages. This increases the risk that adversarial intermediaries may disrupt the protocol by discarding, delaying, or selectively refusing to relay messages, preventing honest parties from completing transactions in a timely manner. We refer to this as the active offline issue. More critically, multi-hop settings involve multiple blockchains, resulting in a greater number of required on-chain interactions and longer execution time. Once an active offline issue occurs, the affected party's funds remain locked on the blockchain for an extended period and cannot be used. Active offline issues arise  before the settlement phase, because once one party settles, the other party can obtain the secret necessary to complete its settlement.  

To solve active offline issues, current solutions attempt to employ premium mechanisms to compensate the victim or shorten the time interval for time locks on both blockchains~\cite{han2019optionality,xue2021hedging,mazumdar2022towards,liu2018atomic,ding2022lilac}. 
However, the first approach fails to address the fundamental issue of unavailable funds, leaving the victim unable to access locked funds. Meanwhile, the second method compromises the security of the interaction. For example, in Interledger~\cite{Interledger}, a default time window is set for completing the protocol.  
Shortening this window would leave one party without sufficient time to finish its required actions. Therefore, the key to solving this issue lies in enabling the other party to access the locked fund without waiting for a timeout when one party goes offline maliciously. %

\noindent \textbf{[C2] Passive offline issues.} 
Passive offline issues occur during the settlement phase. After one party has completed settlement, another party may become offline due to external factors (e.g., network problems, hardware failures, software crashes) and fail to complete its settlement, which causes financial loss.  
Current solutions~\cite{guo2023cross} to the passive offline issue rely on an incentive strategy. 
This strategy is based on the existing design in which the two interacting parties share the same hash lock, which allows honest nodes to detect the offline status of either party. 
However, this approach contradicts the goal of providing unlinkability through the use of independent hash locks~\cite{guo2024zkCross,heilman2017tumblebit,malavolta2017concurrency,deshpande2020privacy,tairi20212}. 
The core challenge in addressing this issue is how to enable honest nodes to assist offline nodes in completing the protocol when using independent hash locks. The difficulty arises from the inability of honest nodes to determine whether a node is truly offline when independent hash locks are employed. 

\noindent \textbf{[C3] Compromising unlinkability.} 
Multi-hop interactions require multiple transactions to ensure atomicity. This increases the risk that adversaries correlate information to compromise unlinkability~\cite{xu2024exploring}. 
Unlinkability refers to the property that prevents adversaries from linking a sender to a receiver with more than negligible advantage. It is often formalized through a security game in which an adversary is given a set of possible sender–receiver pairs and must identify the actual pair based on observed transactions~\cite{guo2023cross,malavolta2017concurrency,deshpande2020privacy,tairi20212}. A system achieves unlinkability if the adversary cannot outperform random guessing. 
While certain existing works aim to achieve cross-chain unlinkability, they often come with limitations. Some solutions are designed for specific use cases~\cite{zhang2021privacy, yi2022ccubi}, while others rely on anonymous blockchains~\cite{li2022zerocross, baldimtsi2021anonymous}, which may not always be available. The mechanism in~\cite{deshpande2020privacy} achieves unlinkability for cross-chain exchanges but assumes that both parties share a secret, which is impractical when dealing with untrusted users. 
zkCross~\cite{guo2024zkCross} removes these assumptions and enables unlinkable cross-chain interactions. 
However, it restricts users to fixed fund amounts, which limits flexibility in meeting diverse user needs.
In addition, several techniques originally developed for single-chain settings~\cite{malavolta2017concurrency,qin2023blindhub,ge2023accio} have attempted to support unlinkability. However, these approaches typically assume that all participants remain online. As discussed in our analysis of passive offline issues {\bf [C2]}, current designs for handling offline scenarios often conflict with unlinkability requirements, making it difficult to achieve both goals within a system. 
%

Therefore, the current state of the art prompts us to consider the following question: 

\textit{
Is it possible to design a cross-chain channel that supports unlinkable multi-hop interactions while also addressing both active and passive offline issues? } 

To address the above problem, we introduce a cross-chain channel network, denoted as $\mathsf{CCN}$.  
The main contributions of this work are summarized below. 

\subsection{Our Contributions}
 CCN adopts off-chain channels to facilitate cross-chain transactions and proposes a novel protocol $\mathsf{R\text{-}HTLC}$ to achieve multi-hop settlements. $\mathsf{R\text{-}HTLC}$ includes a novel \textit{hourglass mechanism} to solve active offline issues, which allows nodes to spend a portion of the locked funds, ensuring that the funds are not continuously unavailable even in scenarios with prolonged time locks. Furthermore, $\mathsf{R\text{-}HTLC}$ adopts zk-SNARK (Zero-Knowledge Succinct Non-interactive Argument of Knowledge)~\cite{sun2021survey} to generate independent hash locks and provides a \textit{multi-path refund strategy} to address passive offline issues. This strategy imposes constraints on nodes to assist offline nodes and ensures the correctness of settlements in different cases caused by offline issues. 
We present a comparison with previous approaches in Table~\ref{tab:compare}, which shows that 
$\mathsf{CCN}$ inherits the efficiency of cross-chain channels in executing cross-chain transactions and employs $\mathsf{R\text{-}HTLC}$ to complete cross-chain settlements. Compared to current privacy-preserving cross-chain works~\cite{baldimtsi2021anonymous,zhang2023cross,guo2024zkCross}, $\mathsf{CCN}$ can support unlinkable interactions under variable funds without relying on third parties or anonymous blockchains. 
\begin{table}[htb] \footnotesize
\centering
\caption{Comparison Among Different Approaches.}
\vspace{-2pt}
\label{tab:compare}
\begin{threeparttable}
\begin{tabularx}{0.995\linewidth}{|>{\centering\arraybackslash}m{1.77cm}|>{\centering\arraybackslash}m{1.41cm}|>{\centering\arraybackslash}m{1.74cm}|>
{\centering\arraybackslash}m{0.73cm}|>
{\centering\arraybackslash}m{0.97cm}|}
\hline
Scheme & Cross-Chain & Active/ Passive Offline Issue  & Privacy & Efficiency \\
\hline
Refs.~\cite{qin2023blindhub,ge2023accio} & $\Circle$ & $\Circle$ /  $\Circle$  & $\CIRCLE$ & $\CIRCLE$ \\
\hline 
Ref.~\cite{guo2023cross} & $\CIRCLE$ & $\Circle$ / $\CIRCLE$  &  $\Circle$  & $\CIRCLE$ \\
\hline 
Refs.~\cite{xue2021hedging,singh20254} & $\CIRCLE$ & $\CIRCLE$ / $\Circle$  & $\Circle$ & $\Circle$ \\
\hline
Refs.~\cite{guo2024zkCross,han2024p2c2t} & $\CIRCLE$ & $\Circle$ / $\Circle$  & $\CIRCLE$ & $\Circle$ \\
\hline 
Refs.~\cite{thyagarajan2022universal,ni2025pipeswap} & $\CIRCLE$ & $\Circle$ / $\Circle$  & $\Circle$ & $\Circle$ \\
\hline

{\bf Our Work} & $\CIRCLE$ & $\CIRCLE$ / $\CIRCLE$ & $\CIRCLE$ & $\CIRCLE$ \\
\hline
\end{tabularx}
\begin{tablenotes} \scriptsize 
\item[*] $\CIRCLE$ indicates that the scheme provides the property or mitigates the issue; $\Circle$ indicates that it does not. 
\end{tablenotes}
\end{threeparttable}
\vspace{-2pt}
\end{table}

For the sake of convenience, we highlight our contributions as follows: 
\begin{enumerate}
\item 

We propose $\mathsf{CCN}$, a cross-chain channel network that addresses the challenges of active and passive offline issues in multi-hop applications while ensuring unlinkable interactions. $\mathsf{CCN}$ inherits the efficiency of off-chain channels and adopts a novel protocol, $\mathsf{R\text{-}HTLC}$, to achieve multi-hop settlements. 

\item As the key block of $\mathsf{CCN}$,  $\mathsf{R\text{-}HTLC}$ include a novel hourglass mechanism to mitigate active offline issues in multi-hop settlements. This mechanism allows nodes to utilize a portion of locked funds, ensuring that liquidity is not entirely frozen even under prolonged time-lock conditions. 

\item 
To address the passive offline issue, we propose a multi-path refund strategy in $\mathsf{R\text{-}HTLC}$. This strategy considers different settlement cases, including normal completion, active offline, and passive offline scenarios. It adopts suitable settlement methods for each case based on their specific conditions, ensuring that the protocol can handle all situations correctly. 

\item In addition to $\mathsf{R\text{-}HTLC}$, $\mathsf{CCN}$ adopts ZKPs to enable unlinkable multi-hop cross-chain interactions. Unlike existing solutions that rely on anonymous blockchains or enforce fixed transaction amounts, $\mathsf{CCN}$ achieves unlinkability while allowing variable funds, making it more practical for real-world applications. 

\item To validate the performance of our scheme, we conduct extensive experiments on both Ethereum and Cosmos. Our evaluation focuses on three aspects: the effectiveness of the protocol in mitigating offline issues; the additional latency and gas overhead introduced by the privacy-preserving design; and the scheme’s applicability in heterogeneous blockchain environments.  
\end{enumerate}

\section{The Model and Preliminaries}~\label{sec:Preliminaries}
%
%
%
In this section, we first introduce our network model, threat model, simulations, and design goals. Then, we provide preliminaries on HTLC, off-chain channels, and ZKPs, which are the basic building blocks of our schemes. 
%
%
\subsection{Models} \label{subsec: model}
\noindent {\bf Network Model.} 
All entities in our scheme refer to real-world users who own blockchain accounts and can engage in off-chain and on-chain operations. 
Furthermore, each entity can possess multiple accounts, with each account equivalent to a node on the blockchain. 
%
All blockchains in our scheme are public chains, meaning that any entity can register nodes on all blockchains and access the transactions within them. 
Two nodes from different entities can participate in intra-chain interactions by establishing off-chain channels (also referred to as channels). 
Note that we do not consider the scenario where a single entity creates two accounts to construct an off-chain channel. 
Select one channel from each of the two chains. When the two nodes in both channels belong to the same entity, those channels can form a cross-chain channel, enabling cross-chain interactions.  
For example, if $u_2$ and $u_3$ from two channels $\Omega_{\langle u_1, u_2 \rangle}$ and $\Omega_{\langle u_3, u_4 \rangle}$ belong to the same entity, then the nodes within these channels can perform cross-chain interactions. For convenience, we refer to $u_1$ and $u_4$ as the starting node and the ending node, respectively, while $u_2$ and $u_3$ are termed intermediary nodes. 
Note that,  
%
a cross-chain channel consists of at least two off-chain channels, including one starting node, one ending node, and multiple intermediary nodes. 
Nodes within the cross-chain channel can all act as senders and perform cross-chain interactions with other nodes. All cross-chain channels form a cross-chain channel network. 

All transactions in the network can be divided into two categories: traditional on-chain transactions (also referred to as transactions), which are confirmed and verified through the blockchain consensus mechanism, and off-chain transactions (also referred to as receipts), which exist within channels and are verified by the nodes in those channels. The final result of the receipts is eventually packaged into on-chain transactions, updating the on-chain states of the nodes within the channel. 

\noindent {\bf Threat Model.}
We assume that all participants operate in probabilistic polynomial time (PPT). Adversaries may exhibit malicious  behaviors, including deviating from protocol specifications, delaying messages, or attempting to link transactions. 
Our protocol provides enhanced robustness by additionally addressing the practical challenge of node offline issues, i.e., it can tolerate the offline failure of a node $u_i$ in a channel $\Omega$. 
Moreover, we follow the standard assumption used in current multi-hop privacy schemes~\cite{malavolta2017concurrency,piotrowska2017loopix}, namely that adversaries may corrupt an arbitrary subset of intermediary nodes, under the minimal requirement that the entire path is not fully controlled by the adversary. This assumption is consistent with the threat models adopted in practical anonymous communication systems such as Tor, Nym, and Loopix~\cite{dingledine2004tor,nym}. 
We also consider consensus-layer threats. For instance, in PoW-based blockchains, the adversary cannot control more than 50\% of the hash power~\cite{gervais2016security}, and in PBFT-based systems, fewer than one-third of nodes may be Byzantine~\cite{xiao2020survey}. These assumptions ensure the liveness and consistency of the underlying ledgers. 

\noindent{\bf The Simulation of Offline Issues.} 
To demonstrate the presence of offline vulnerabilities, i.e., active and passive offline issues, we simulate their occurrence during the execution of the HTLC protocol, as illustrated in Fig.~\ref{Fig:experiments-offline-1}. The figure illustrates two representative scenarios. In the first scenario, after one node completes its locking operation as specified by the protocol, the counterparty intentionally goes offline before the timelock expires, refusing to perform its expected locking operation. This deviation from the protocol prevents the honest node from proceeding to claim its funds, causing them to remain locked and inaccessible for an extended period. In the second scenario, after both parties have locked their funds as required by the protocol, one party proceeds to unlock its portion successfully. However, the other party remains offline and fails to submit the corresponding unlock transaction before the timeout. This leaves the offline party unable to retrieve its entitled funds: it cannot unlock the counterparty’s funds due to timeout, nor can it refund its own locked funds, resulting in a direct and unrecoverable financial loss. 
These experimental results highlight that offline behaviors can directly compromise the correctness and fairness of cross-chain interactions. 

To address these challenges, our protocol is designed to be resilient against both active and passive offline issues. The corresponding design goals are formalized as follows. 
\begin{figure}[htb]
    \centering
    \subfigure[{\bf Active offline issue}]{\includegraphics[width=0.23\textwidth]{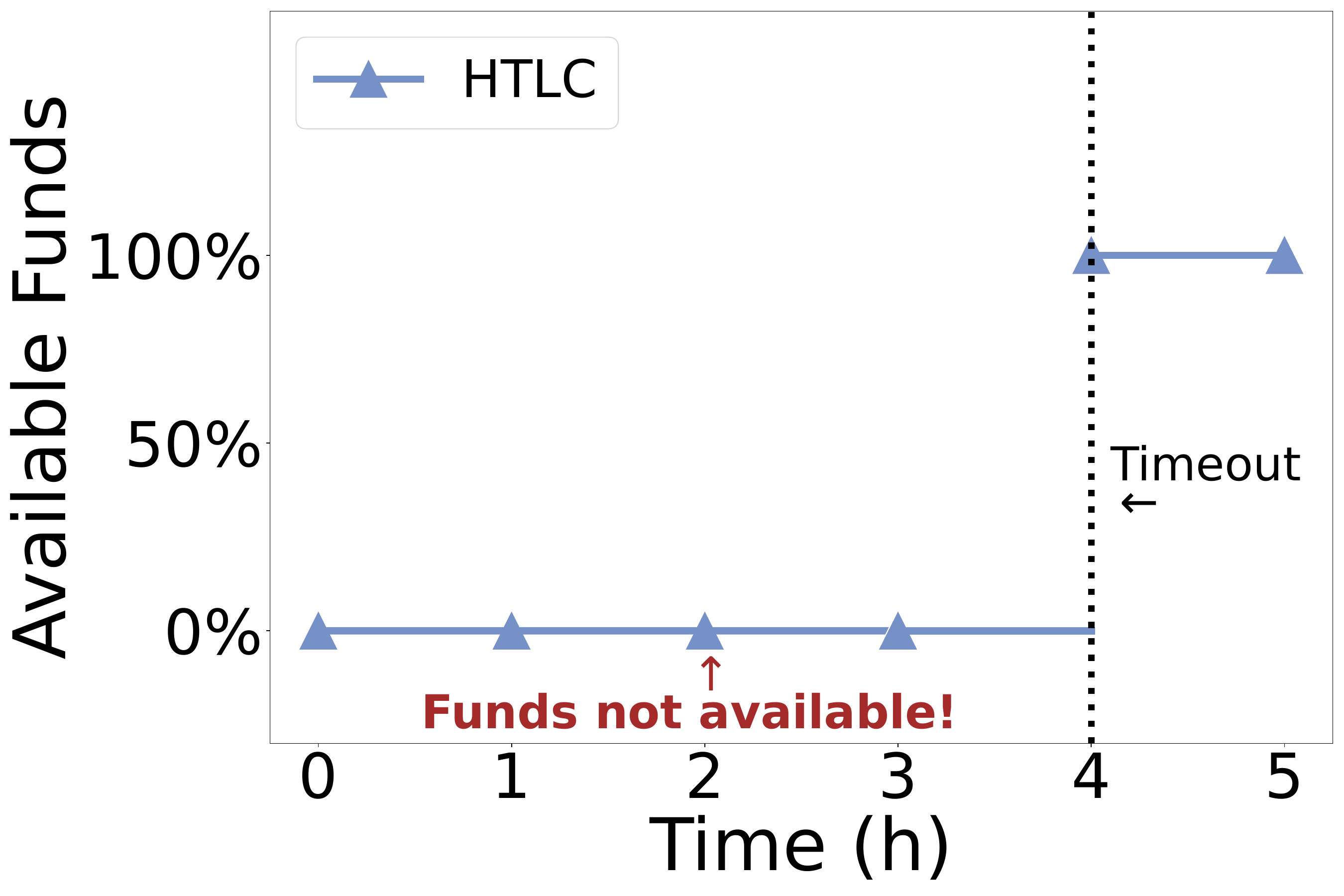}}
    \subfigure[{\bf Passive offline issue}]{\includegraphics[width=0.23\textwidth]{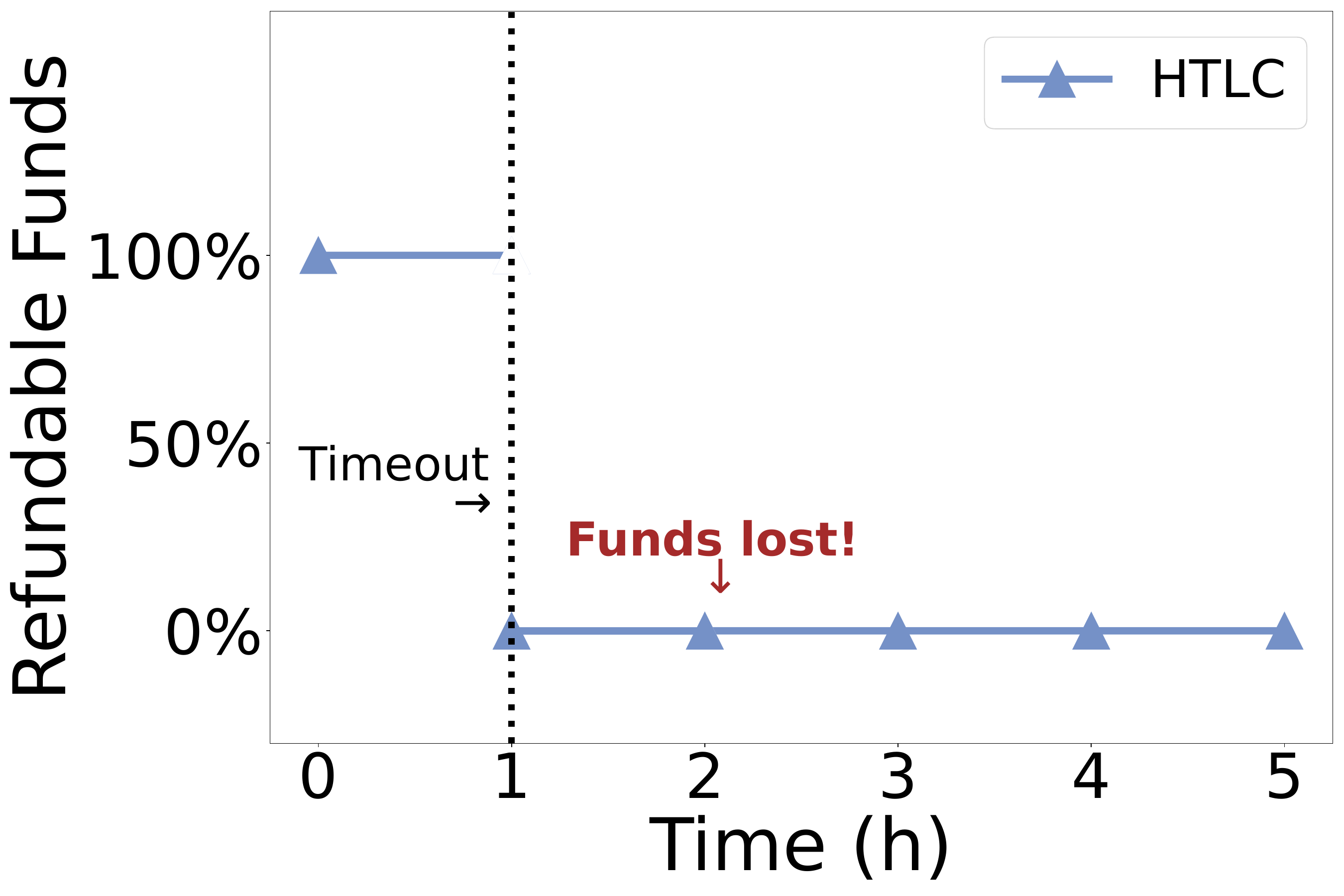}}
    \caption{Simulation of active offline issue (a) and passive offline issue (b). } 
    \label{Fig:experiments-offline-1}
\end{figure}

\noindent {\bf Design Goals.}
In multi-hop cross-chain interactions, a  path consists of multiple off-chain channels, each connecting two adjacent participants. We observe that security vulnerabilities may arise if one party in any of these channels becomes unavailable (offline) during critical stages of the protocol execution.

We identify two categories of such offline-related issues: 
\begin{itemize}
    \item \textbf{Active Offline.} In any two-party channel $\Omega(u_i, u_j)$, an active offline event occurs when, after $u_i$ has completed the locking transaction, the counterparty $u_j$ intentionally goes offline instead of completing its corresponding locking transaction. This results in an incomplete protocol state in which $u_i$’s funds may remain locked until the timelock expires. 
    \item \textbf{Passive Offline.} In any two-party channel $\Omega(u_i, u_j)$, a passive offline event occurs when both parties have completed the locking operations, but after $u_i$ performs the unlock operation, the counterparty $u_j$ becomes offline due to an unexpected failure and is unable to perform its own unlock operation. In such a state, $u_j$ may suffer financial loss unless the protocol provides adequate fault-tolerance.  
\end{itemize} 
%
To ensure the security of cross-chain channel protocols, it is necessary to explicitly account for the impact of offline behavior during protocol execution. Since node unavailability may affect both protocol completion and the safety of locked funds, we formalize the corresponding security requirements below.   

\begin{definition}[Security Against Offline Issues] 
A protocol $\Pi$ is secure against offline issues if, for any PPT adversary $\mathcal{A}$ and any execution of $\Pi$, the following holds except with negligible probability in the security parameter: 
\begin{enumerate} 
    \item If the adversary forces a party offline after one or both sides have completed their locking operations, the honest party must still retain the ability to access and use its locked funds without having to wait until the prescribed timeout.  
    \item If an honest party completes the unlock procedure while its counterparty is offline, then upon returning online the offline party must not incur any financial loss as a result of having been offline. 
\end{enumerate}
\end{definition}

\begin{itemize}
    \item \textbf{Privacy.} Privacy in cross-chain systems is commonly captured by unlinkability,
    which requires that an adversary cannot determine whether a given sender is linked to its intended receiver from observable protocol executions~\cite{malavolta2017concurrency,ge2023accio,guo2024zkCross}. 
    Unlinkability is formalized through the following  game. 
\end{itemize}

\begin{definition}[Unlinkability Game] \label{Def:unlinkability}
Let $\lambda$ be the security parameter. The unlinkability game $\mathsf{Game}^{\mathsf{Unlink}}_{\mathcal{A}}(\lambda)$ is defined as follows:
\begin{enumerate}
    \item The challenger initializes the system and selects two distinct sender-receiver pairs $(S_0, R_0)$ and $(S_1, R_1)$.
    \item A random bit $b \in \{0,1\}$ is chosen, and the challenger simulates a payment from $S_b$ to $R_b$ via a multi-hop path, producing the on-chain and observable outputs $\mathcal{O}$.
    \item Following the threat model, the adversary $\mathcal{A}$, having control over an arbitrary subset of intermediate nodes and observing $\mathcal{O}$, outputs a guess $b'$.
\end{enumerate}
The advantage of $\mathcal{A}$ is defined as:
\[
\mathsf{Adv}^{\mathsf{Unlink}}_{\mathcal{A}}(\lambda) = \left| \Pr[b' = b] - \frac{1}{2} \right|. 
\]

We say that $\mathsf{CCN}$ satisfies \emph{unlinkability} if for all PPT adversaries $\mathcal{A}$, the advantage $\mathsf{Adv}^{\mathsf{Unlink}}_{\mathcal{A}}(\lambda)$ is negligible in $\lambda$. 

\end{definition}

Beyond security and privacy, a cross-chain protocol must guarantee atomic settlement
to ensure the correctness of the final outcome~\cite{guo2024zkCross}. 
Without atomicity, attacks such as double-spending or selective abortion may lead to inconsistent states across blockchains and compromise the correctness of cross-chain settlement.  
\begin{definition}[Cross-Chain Atomicity]
\label{Def:atomicity}
A cross-chain scheme satisfies \emph{atomicity} if, for any execution involving multiple blockchains, either all honest participants obtain their agreed outputs, or all honest participants can reclaim their locked assets. 
\end{definition}

In the rest of this section, we will introduce the main three key technologies that we adopt to achieve these design goals.  

\subsection{HTLC} ~\label{subsec:HTLC} 
HTLC initially developed within the Lightning Network, has been adapted for cross-chain applications to guarantee atomicity in cross-chain interactions~\cite{xu2024exploring}.  
\begin{definition}[HTLC in Cross-chain Scenarios]\label{Def:HTLC}
  The entire process involves two blockchains $\alpha$ and $\beta$, where both interacting parties Alice and Bob have nodes on both chains. The protocol comprises three steps: $\mathsf{Lock}$, $\mathsf{Unlock}$, and $\mathsf{Refund}$.  
\begin{itemize}
        \item  $\mathsf{Lock}$. Alice generates a preimage $s$, computes its hash result $h(s)$, and sends $h(s)$ to Bob. In chain $\alpha$, Alice employs $h(s)$ and a time lock $T$ to lock the funds sent to Bob. Similarly, Bob uses the same method in chain $\beta$ to lock the funds sent to Alice. 
        \item  $\mathsf{Unlock}$. In chain $\beta$, Alice opens $s$ to unlock the funds. Then, Bob learns $s$ and unlocks the funds in chain $\alpha$.  
        \item $\mathsf{Refund}$. If Alice fails to provide $s$ before timeout, the locked funds in chain $\beta$ would be refunded to Bob, and the locked funds in chain $\alpha$ would be refunded to Alice. 
\end{itemize}
\end{definition}

Note that, the interval between these two time locks ensures that after Alice completes the unlock process in chain $\beta$, Bob has sufficient time to complete the unlock process in chain $\alpha$. 

\subsection{Off-Chain Channels} ~\label{subsec:channel}
The off-chain channel serves as a scalability solution for blockchains. 
\begin{definition}[Off-chain Channel]
   The whole process of off-chain channels $\Theta$ can be can be divided into three steps: $\mathsf{Open}$, $\mathsf{Interact}$, $\mathsf{Close}$.  
    \begin{itemize}
        \item  $\mathsf{Open}$.  Both parties open a channel by depositing funds into the smart contract, establishing the initial state of the channel. 

        \item  $\mathsf{Interact}$. Both parties interact within the channel using receipts in an off-chain manner. Each receipt includes the sender's signature, the receiver's address, and the transfer funds. Both parties need to execute each receipt to update the channel state, and the updated state requires signatures from both. Note that, every state change corresponds to a sequence number, which records the order of each state change. 

        \item  $\mathsf{Close}$. After completing their interactions, one party (agreed by both parties) uploads the final state to close the channel. If the uploading party acts maliciously by providing a biased state, the other party can appeal within a set time frame with the correct one. The smart contract would authenticate the final state by verifying the signatures and sequence numbers provided by both parties. 
    \end{itemize}
\end{definition}
%
The core design of off-chain channels allows multiple transactions to be conducted off the blockchain, thereby alleviating on-chain congestion.

\subsection{Zero-knowledge Proof: zk-SNARK} ~\label{subsec: zk-SNARK}
zk-SNARK is one type of zero-knowledge proof widely utilized in blockchain applications~\cite{sun2021survey}. %
\begin{definition}[zk-SNARK]
   The whole process of zk-SNARK can be represented by a tuple of polynomial-time algorithms $\Pi$ $\overset{\text{def}}{=}$ ($\mathsf{Setup}$, $\mathsf{Prove}$, $\mathsf{Verify}$): 
		\begin{itemize}
			\item  $(\widetilde{pk}, \widetilde{vk}) \leftarrow \mathtt{Setup}(1^{\lambda}, \Lambda)$. Taking a security parameter $1^{\lambda}$ and a circuit $\Lambda$ as inputs, the algorithm obtains the key pair $(\widetilde{pk}, \widetilde{vk})$, where $\widetilde{pk}$ is the proving key for proof generation, and $\widetilde{vk}$ is the verification key for proof validation. 
			\item $\pi \leftarrow \mathtt{Prove}(\widetilde{pk}, \vec{x}, \vec{w})  $.
            The algorithm takes the proving key $\widetilde{pk}$, the public inputs $\vec{x}$, and the private inputs $\vec{w}$ as inputs to generate a proof $\pi$. 
			\item $ 1/0   \leftarrow  \mathtt{Verify}( \widetilde{vk}, \vec{x}, \pi)$.
			The algorithm validates the proof $\pi$ using the verification key  $\widetilde{vk}$ and public inputs $\vec{x}$, returning a result of 1 for successful verification and 0 otherwise. 
		\end{itemize}
  \end{definition}
%
zk-SNARK satisfies the properties of soundness, completeness, and zero knowledge.  

In the following, we will introduce how to design our scheme based on the above technologies. 


\section{CCN: Cross-Chain Channel Networks}~\label{sec:CCN}
%
%
%
In this section, we first provide an overview of $\mathsf{CCN}$. Then, we introduce our newly proposed cross-chain settlement protocol $\mathsf{R\text{-}HTLC}$, and explain how it enables $\mathsf{CCN}$ to support cross-chain interactions. 
\subsection{Overview}
\begin{figure}[ht]
    \centering
    \includegraphics[width=0.47\textwidth]{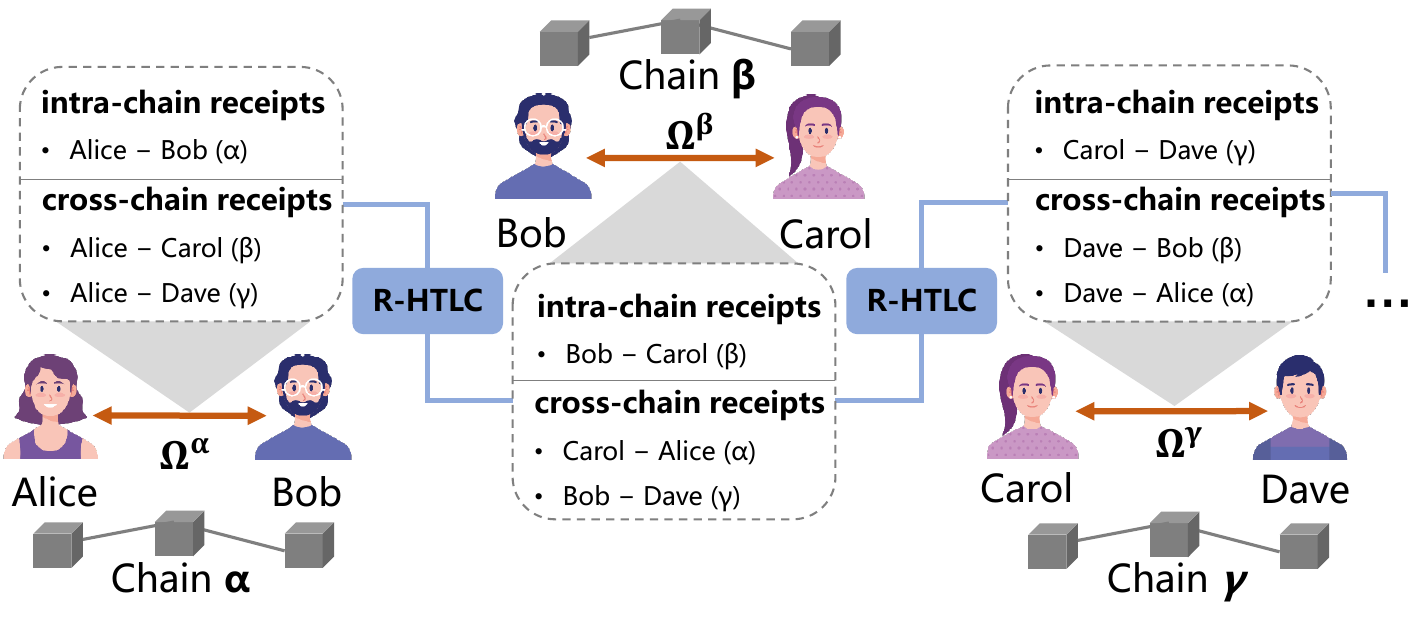}
	\caption{This illustration provides an overview of the cross-chain channel network $\mathsf{CCN}$. The orange arrows represent off-chain channels, and cross-chain settlement between different blockchains is facilitated via $\mathsf{R\text{-}HTLC}$. 
    } 
    \label{Fig:overview}
\end{figure}
$\mathsf{CCN}$ is designed to establish a cross-chain channel network by connecting multiple cross-chain channels, where each cross-chain channel consists of a pair of off-chain channels anchored on different blockchains. In this network, nodes can perform cross-chain interactions by routing through intermediate nodes. For example, as illustrated in Fig.~\ref{Fig:overview}, Alice and Bob interact via an off-chain channel on blockchain $\alpha$, Bob and Carol communicate on blockchain $\beta$, and Carol and Dave are connected via blockchain $\gamma$. This structure naturally extends to include other blockchains and participants, forming a decentralized cross-chain channel network. 

In $\mathsf{CCN}$, the interacting parties exchange signed receipts to serve as off-chain evidence of the interaction (introduced in Sec.~\ref{subsec:channel}). We further distinguish receipts as intra-chain receipts and cross-chain receipts to indicate the different purposes of the funds. 
Intra-chain receipts are used when both parties are anchored on the same blockchain, whereas cross-chain receipts are used to bridge interactions across different blockchains.  
To ensure secure and atomic settlement across chains, $\mathsf{CCN}$ introduces the protocol $\mathsf{R\text{-}HTLC}$, which guarantees the correctness of settlement between different channels. $\mathsf{R\text{-}HTLC}$ is designed to address both active and passive offline issues. In addition to $\mathsf{R\text{-}HTLC}$, $\mathsf{CCN}$ employs ZKPs to preserve unlinkability by preventing the leakage of information that may otherwise reveal transaction linkages. 

We present the detailed design of $\mathsf{R\text{-}HTLC}$ and $\mathsf{CCN}$ in the following subsections. 

\subsection{R-HTLC: Revisiting HTLC in Light of Security Concerns}~\label{sec:R-HTLC} 
%
In this subsection, we introduce $\mathsf{R\text{-}HTLC}$, and provide its definition in the following:  
\begin{definition}
    $\mathsf{R\text{-}HTLC}$ is a cross-chain settlement protocol, which can be represented by a tuple of polynomial-time algorithms $\mathsf{R\text{-}HTLC}$ $\overset{\text{def}}{=}$ ($\mathsf{Prepare}$, $\mathsf{Lock}$, $\mathsf{Unlock}$, $\mathsf{Refund}$).  
\end{definition}

$\mathsf{R\text{-}HTLC}$ enhances the security of HTLC and has three key technological improvements.  
First, in $\mathsf{R\text{-}HTLC}.\mathsf{Prepare}$, we adopt zk-SNARK to generate multiple independent hash locks suitable for multi-hop cross-chain scenarios, thereby addressing the privacy risks associated with using the same hash locks. 
Second, we propose an hourglass mechanism in $\mathsf{R\text{-}HTLC}.\mathsf{Lock}$ to solve the active offline issue, where nodes can spend a portion of the locked funds. 
Finally, in $\mathsf{R\text{-}HTLC}.\mathsf{Refund}$, we design a multi-path refund mechanism that handles different settlement outcomes, including normal completion as well as refund scenarios arising from offline events. 
%
%
%
%
%
%
%

%

We consider the following scenario: Alice intends to initiate a cross-chain transfer of 30 units to Carol, and a cross-chain channel from Alice to Carol is formed through an intermediary node Bob. Note that in real-life scenarios, different chains have publicly known exchange rates. Therefore, we assume an exchange rate of 1:2, meaning that if Alice transfers 30 in Chain~$\alpha$ to Bob, then Bob would need to transfer 60 units to Carol in Chain~$\beta$. 
Each account has a key pair ($pk$, $vk$) to perform on-chain operations, such as sending transactions and verifying the validity of transaction signatures. 

\subsubsection{The Prepare Phase Based on zk-SNARKs} 
This phase focuses on generating the circuit required for the locking phase. To support later phases aimed at addressing offline issues, $\mathsf{R\text{-}HTLC}$ introduces three hash locks. In current HTLC-based schemes, a single hash lock is typically reused across multiple hops to enable sequential unlocking and ensure atomicity. However, this design can compromise unlinkability, as identical hash values can be observed across transactions. Therefore, the main challenge is to conceal these hash locks while still enabling honest participants to perform the sequential unlocking process correctly. %

To generate multiple independent hash locks, we construct two new circuits in Fig.~\ref{Fig:prepare-circuit} and employ zk-SNARK to protect the private inputs. 
First, Alice generates two circuits, i.e., $\Lambda^{\mathsf{Pre.}}_{CB}$ and $\Lambda^{\mathsf{Pre.}}_{CA}$. 
\begin{figure}[htb]
	\centering
	\includegraphics[width=0.4\textwidth]{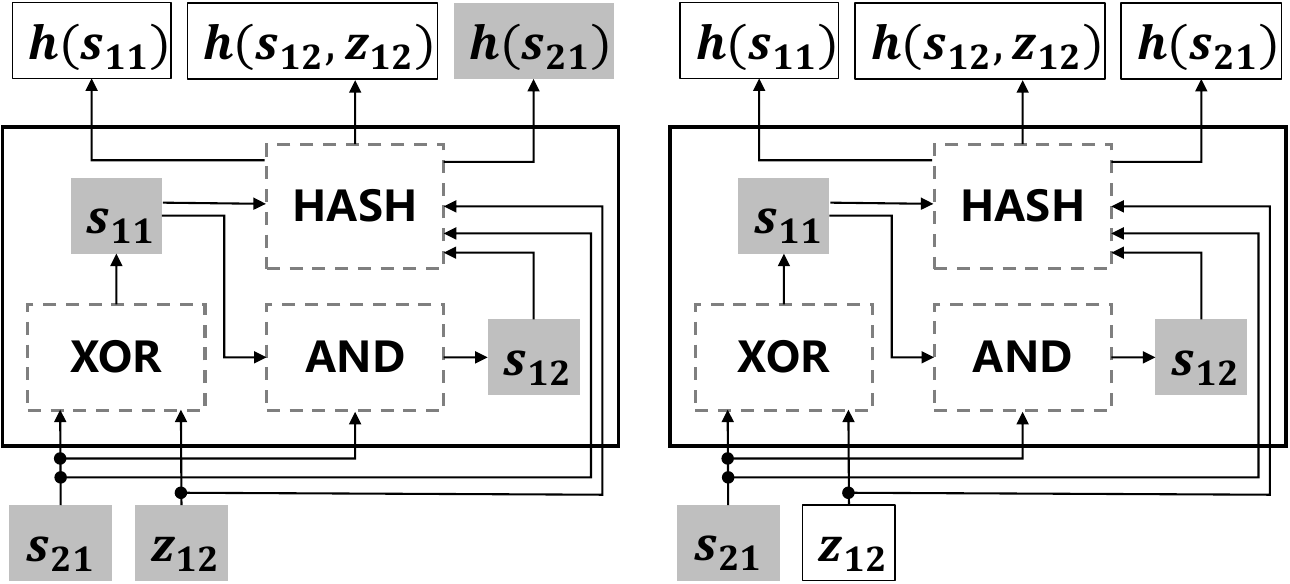}
	\caption{ The logic diagram of the circuits in $\mathsf{R\text{-}HTLC.Prepare}$ -- the circuit on the left is denoted as $\Lambda^{\mathsf{Pre.}}_{CA}$, while the one on the right is denoted as $\Lambda^{\mathsf{Pre.}}_{CB}$.} 
\label{Fig:prepare-circuit}
\end{figure}
Both of these circuits incorporate three functions (indicated by dashed lines): $\mathsf{AND}$ (Logical AND), $\mathsf{HASH}$ (SHA-256), and $\mathsf{XOR}$ (eXclusive OR).  
They aim to generate three hash locks ($h(s_{11}),$ $h(s_{12},$ $z_{12}),$ $h(s_{21})$) and establish the operational relationships among three 256-bit preimages ($s_{11},s_{12},s_{21}$) with a 256-bit integer ($z_{12}$). 
The specific constraints are summarized as follows\footnote{All parameters with a gray background are private ones protected with zk-SNARK.}: 
\begin{itemize} 
\item [--] $\Lambda^{\mathsf{Pre.}}_{CA}$:  \colorbox{gray!50}{$s_{11}$} = \colorbox{gray!50}{$s_{21}$} $\oplus$ \colorbox{gray!50}{$z_{12}$}; \colorbox{gray!50}{$s_{12}$} = \colorbox{gray!50}{$s_{11}$} $\land$ \colorbox{gray!50}{$s_{21}$}; 
\item [--] $\Lambda^{\mathsf{Pre.}}_{CA}$: $h(s_{11})$ = SHA-256(\colorbox{gray!50}{$s_{11}$}); $h(s_{12},z_{12})$ = SHA-256(\colorbox{gray!50}{$s_{12},z_{12}$}); $h(s_{21})$ = SHA-256(\colorbox{gray!50}{$s_{21}$}). 
\item [--] $\Lambda^{\mathsf{Pre.}}_{CB}$: \colorbox{gray!50}{$s_{11}$} = \colorbox{gray!50}{$s_{21}$} $\oplus$ $z_{12}$; \colorbox{gray!50}{$s_{12}$} = \colorbox{gray!50}{$s_{11}$} $\land$ \colorbox{gray!50}{$s_{21}$}; 
\item [--] $\Lambda^{\mathsf{Pre.}}_{CB}$: $h(s_{11})$ = SHA-256(\colorbox{gray!50}{$s_{11}$}); $h(s_{12},z_{12})$ = SHA-256(\colorbox{gray!50}{$s_{12}$}, $z_{12}$); \colorbox{gray!50}{$h(s_{21})$} = SHA-256(\colorbox{gray!50}{$s_{21}$}). 
\end{itemize}
%
%
%
Note that, each circuit, including those introduced later, must undergo an initialization process to facilitate the generation of ZKPs, and we use circuit $\Lambda^{\mathsf{Pre.}}_{CB}$ as an example to illustrate this process. 
To be specific, Carol first generates a security parameter and employs $\Pi.\mathsf{Setup}$ to derive the proving key and the verification key ($\widetilde{pk}^{\mathsf{Pre.}}_{CB}$, $\widetilde{vk}^{\mathsf{Pre.}}_{CB}$). 
Then, Carol adopts the algorithm $\Pi.\mathsf{Prove}$ to generate a proof $\pi^{\mathsf{Pre.}}_{CB}$, with the proving key $\widetilde{pk}^{\mathsf{Pre.}}_{CB}$, the public inputs ($h(s_{11}),h(s_{12},z_{12}),h(s_{21}),z_{12}$), and the private inputs ($s_{11},s_{12},s_{21}$).   
The other circuits follow the same steps as described above, and we would omit this process in subsequent descriptions. So, for the circuit $\Lambda^{\mathsf{Pre.}}_{CA}$, Carol would generate a proof $\pi^{\mathsf{Pre.}}_{CA}$. 
Finally, through off-chain communications, Carol sends $\pi^{\mathsf{Pre.}}_{CB}$ to Bob and $\pi^{\mathsf{Pre.}}_{CA}$ to Alice, along with their respective public inputs. Bob and Alice use $\Pi.\mathsf{Verify}$ to verify the correctness of the received proof. 

\subsubsection{Locking with an Hourglass Mechanism}
To address the issue of fund unavailability caused by active offline behavior, we design an hourglass mechanism that allows honest parties to access and use their locked funds without waiting for the timeout.  

In this phase, one first locks the funds on the smart contract and then effectively spends these funds based on an hourglass mechanism. 
Specifically, Alice sends a transaction $\mathsf{Tx}_{\text{Lock}}$ to the smart contract $\xi^\alpha$, which aims to lock the fund $v_\mathsf{{Lock}}$ (i.e., 30 uints) based on hash locks $h(s_{11})$ and $h(s_{12},z_{12})$. Then, $\xi^\alpha$ sets a time lock $T_1$. Bob can provide either $s_{11}$ or ($s_{12}$, $z_{12}$) within $T_1$ to unlock $v_\mathsf{{Lock}}$. 
After that, $\xi^\alpha$ initiates an hourglass mechanism.
The hourglass mechanism segregates the funds locked in the smart contract into frozen funds and available funds. The frozen funds originate from $\mathsf{Tx}_{\text{Lock}}$ sent by Alice, with available funds initially set to 0. 
As time goes by, frozen funds gradually transform into available funds, allowing the nodes who locked the funds to generate multiple commitments equal to or less than the available funds for payment purposes. 
For example, 
when the available funds exceed 0, Alice can choose to generate a commitment $\mathsf{cmt}$ similar to checks in banks. 
Let $\mathsf{cmt} \overset{\text{def}}{=} (\mathsf{snd}, \mathsf{rcv}, v_\mathsf{cmt})_{\mathsf{\sigma_{\xi}}}$, meaning that a sender $\mathsf{snd}$ transfers $v_\mathsf{cmt}$ to a receiver $\mathsf{rcv}$, where $\mathsf{\sigma_{\xi}}$ is the signature signed by the smart contract $\xi$. 
Alice sends ($\mathsf{rcv}, v_\mathsf{cmt}$) based on $\mathsf{Tx}_{\text{Check}}$ to $\xi^\alpha$. $\xi^\alpha$ returns $\mathsf{cmt}$ to Alice when it verifies that the $v_\mathsf{cmt}$ is less than or equal to the available funds. Then, Alice sends this $\mathsf{cmt}$ to $\mathsf{rcv}$. 
$\mathsf{rcv}$ can get the corresponding funds during $\mathsf{R\text{-}HTLC.Refund}$. 
Note that the rate at which frozen funds convert to available funds is determined by the interacting parties. 
If both parties trust each other, they can set it to 0. 

After Alice completes the locking process, Bob in Chain $\beta$ sends $\mathsf{Tx}_{\text{Lock}}$ to the smart contract $\xi^\beta$ to lock the funds intended for Carol using the $h(s_{21})$. 
However, due to potential delays between the time Alice and Bob submit $\mathsf{Tx}_{\text{Lock}}$, and under the effect of the hourglass mechanism, the amount of funds transferred from Alice to Bob may change. For instance, the originally frozen amount of 30 uints on Chain~$\alpha$ may be adjusted to 28 uints. 
As a result, Bob only needs to lock 56 uints on Chain~$\beta$ when performing the corresponding lock operation. 
Then, $\xi^\beta$ sets a time lock $T_2$ (less than $T_1$) and initiates the hourglass mechanism. Using this approach, if either party goes offline during this phase, it would not result in the other party's funds being unavailable for an extended period. 
%

\subsubsection{The Unlock Phase in R-HTLC}
%
Similar to the consideration of privacy in the prepare phase, HTLC-based schemes require revealing the same secret during the unlocking process, which can compromise unlinkability. To address this, we design dedicated unlocking circuits and adopt ZKPs to hide the secret information. 
The unlocking process is divided into two steps. First, both parties within the same off-chain channel confirm the fund to be unlocked. Second, we design two circuits and adopt zk-SNARK to unlock the confirmed fund. 
Due to the effects of the hourglass mechanism, the frozen funds would decrease. Let's assume that the frozen funds are 27 in Chain~$\alpha$ and 54 in Chain~$\beta$. 
The details of the first step are as follows: 

First, 
Bob generates a hash digest $h(h(s_{11}), 27)$, where $h(s_{11})$ is a hash lock in Chain~$\alpha$ and 27 is the fund Bob should withdraw in Chain $\alpha$. 
Second, Bob signs $h(h(s_{11}), 27)$ (marked as $\sigma^{\alpha}_{BA}$) using his private key $vk^\alpha_B$. 
Then, Bob sends $\sigma^{\alpha}_{BA}$ along with $h(h(s_{11}), 27)$, $h(s_{11})$, and 27 to Alice. 
Finally, Alice validates $\sigma^{\alpha}_{BA}$, signs $h(h(s_{11}), 27)$ (marked as $\sigma^{\alpha}_{AB}$), and sends $\sigma^{\alpha}_{AB}$ back to Bob. 
Repeating the above process, Bob and Carol can also confirm the final fund of 54. 
The above process occurs off-chain and ensures that participants within the same channel, such as Alice and Bob, confirm the final fund and exchange signatures with each other regarding the agreed-upon fund. 
\begin{figure}[htb]
	\centering
	\includegraphics[width=0.4\textwidth]{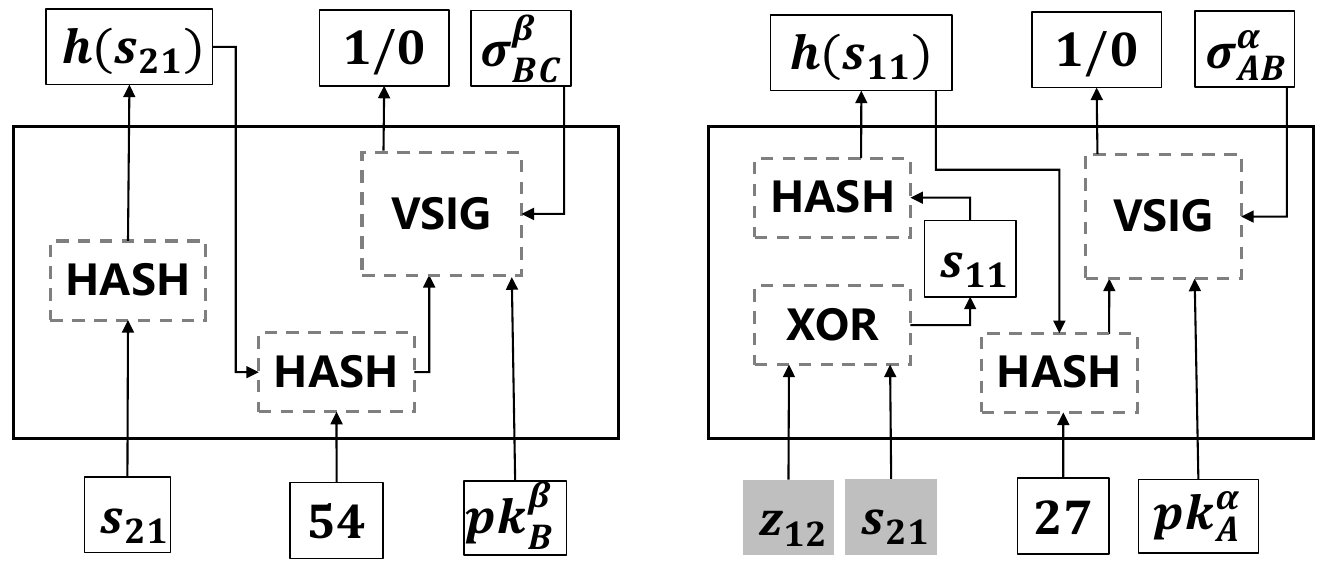}
	\caption{ 
     The logic diagram of the circuits in $\mathsf{R\text{-}HTLC.Unlock}$ -- the circuit on the left is denoted as $\Lambda^{\mathsf{Unl.}}_{C\xi}$, while the one on the right is denoted as $\Lambda^{\mathsf{Unl.}}_{B\xi}$. 
 } 
    \label{Fig:unlock-circuit}
\end{figure}
%
%
%
%
%
In the second step, Carol constructs a new circuit $\Lambda^{\mathsf{Unl.}}_{C\xi}$ shown in Fig.~\ref{Fig:unlock-circuit}. 
The entire circuit includes two types of functions (represented by dashed boxes): $\mathsf{HASH}$ and $\mathsf{VSIG}$, where $\mathsf{VSIG}$ implements the verification of a given signature. The constraints of $\Lambda^{\mathsf{Unl.}}_{C\xi}$ are as follows: 
\begin{itemize} 
\item [--] $\Lambda^{\mathsf{Unl.}}_{C\xi}$: $h(s_{21})$ = SHA-256($s_{21}$); $h(h(s_{21}), 54)$ = SHA-256($h(s_{21})$, 54); 
\item [--] $\Lambda^{\mathsf{Unl.}}_{C\xi}$: 1/0 = VSIG($\sigma^{\beta}_{BC}$, $pk^\beta_B$, $h(h(s_{21}), 54)$). 
\end{itemize}
Based on this circuit, Carol gets the proof $\pi^{\mathsf{Unl.}}_{C\xi}$. 
Then, she packages $\pi^{\mathsf{Unl.}}_{C\xi}$ with public parameters 
($h(s_{21})$, $s_{21}$, 54, $\sigma^{\beta}_{BC}$, $pk^\beta_B$) into the transaction $\mathsf{Tx_{\text{Unlock}}}$, and sends it to the smart contract $\xi^\beta$ to acquire the assets of equal value to 54. 
$\xi^\beta$ verifies the validity of $\pi^{\mathsf{Unl.}}_{C\xi}$ through algorithm $\Pi.\mathsf{Verify}$ and transfers 54 uints to Carol. 
After Carol successfully unlocks, Bob learns $s_{21}$ and generates a new circuit $\Lambda^{\mathsf{Unl.}}_{B\xi}$. Compared to $\Lambda^{\mathsf{Unl.}}_{C\xi}$, $\Lambda^{\mathsf{Unl.}}_{B\xi}$ adds an $\mathsf{XOR}$ function. We summarize the constraints of $\Lambda^{\mathsf{Unl.}}_{B\xi}$ as follows: 
\begin{itemize} 
\item [--] $\Lambda^{\mathsf{Unl.}}_{B\xi}$: $s_{11}$ = \colorbox{gray!50}{$s_{21}$} $\oplus$ \colorbox{gray!50}{$z_{12}$}; 
\item [--] $\Lambda^{\mathsf{Unl.}}_{B\xi}$: $h(s_{11})$ = SHA-256($s_{11}$); $h(h(s_{11}), 27)$ = SHA-256($h(s_{11})$, 27); 
\item [--] $\Lambda^{\mathsf{Unl.}}_{B\xi}$: 1/0 = VSIG($\sigma^{\alpha}_{AB}$, $pk^\alpha_A$, $h(h(s_{11}), 27)$). 
\end{itemize}
Following a process similar to Carol's unlocking procedure, Bob generates a proof $\pi^{\mathsf{Unl.}}_{B\xi}$ 
and sends $\mathsf{Tx}_{\text{Unlock}}$ to the smart contract $\xi^\alpha$. After verifying the correctness of $\pi^{\mathsf{Unl.}}_{B\xi}$, $\xi^\alpha$ transfers 27 uints to Bob. 
%

%
Based on the above design, the unlocking process can be successfully completed without compromising unlinkability. However, this process may still be influenced by node offline events, which can lead to situations where one party completes the unlock operation while the other, being offline, is unable to do so, thereby resulting in potential fund loss. To ensure atomic settlement under such conditions, we introduce a multi-path refund mechanism.

\subsubsection{Refunding with a Multi-Path  Strategy}
We propose a multi-path refund strategy to ensure correct final settlement under various cases. Unlike HTLC-based schemes that simply return all locked funds in the refund phase, the refund process in $\mathsf{R\text{-}HTLC}$ must additionally account for factors such as node offline status, the use of commitments, and the influence of the hourglass mechanism. 
We list three cases in this phase and illustrate how the proposed multi-path refund strategy effectively handles each of them. 

In Case 1, Carol and Bob complete the unlock operation, and subsequently, Alice and Bob perform the refund operation due to a timeout. This case illustrates a successful cross-chain interaction, wherein the available funds in the smart contract are less than the initially locked fund. In Case 2, Carol and Bob do not execute the unlock operation, leading to Alice and Bob initiating the refund operation. Case 2 presents a scenario susceptible to active offline issues. In this case, the available funds equal the initially locked fund. In Case 3, Carol completes the unlock operation, but Bob fails to do so, resulting in a situation where Alice may initiate a malicious refund after the timeout. Specifically, Carol successfully unlocks and obtains the funds locked by Bob; however, after the timeout, since Bob did not complete his own unlock operation, Carol is able to refund the funds she originally locked. Case 3 illustrates a situation in which Bob faces passive offline issues. In this scenario, the available funds in Chain $\alpha$ match the initially locked fund, while in Chain $\beta$, the available funds are less than the initially locked fund. 

$\mathsf{R\text{-}HTLC}$ provides a multi-path refund strategy, ensuring the correctness of settlements for the three possible cases. 
In case 1, In Case 1, the participant initiates a refund through $\mathsf{Tx}_{\text{Refund}_1}$. The transaction $\mathsf{Tx}_{\text{Refund}_1}$ specifies the amount to be refunded. For example, Alice on Chain~$\alpha$ submits $\mathsf{Tx}_{\text{Refund}_1}$ to the smart contract. The contract first transfers the amount requested in $\mathsf{Tx}_{\text{Check}}$ to the corresponding node, and then transfers the remaining available funds back to Alice. 
In Case 2, 
Alice finalizes the settlement with ($s_{12},z_{12}$) by sending $\mathsf{Tx}_{\text{Refund}_2}$. When receiving $\mathsf{Tx}_{\text{Refund}_2}$, the smart contract not only validates the signature and fund but also computes the hash value of ($s_{12},z_{12}$) to determine if it matches the hash lock $h(s_{12},z_{12})$. 
Then, the smart contract sets a time lock $T_3$. 
After $T_3$, the refund process for the fund is the same as in Case 1. 

In case 3, Alice has two options: first, to send $\mathsf{Tx}_{\text{Refund}_3}$ to assist Bob in completing the settlement, or second, to act maliciously as if the unlock did not occur in Chain $\beta$ and send $\mathsf{Tx}_{\text{Refund}_2}$ to claim the funds meant for Bob. 
For the first option, Alice needs to generate a circuit similar to $\Lambda^{\mathsf{Unl.}}_{B\xi}$ (shown in Fig.~\ref{Fig:unlock-circuit}) and replace public inputs ($\sigma^{\alpha}_{AB}$, $pk^\alpha_A$) with ($\sigma^{\alpha}_{BA}$, $pk^\alpha_B$) when generating the proof. After that, Alice sends the proof with the public input based on $\mathsf{Tx}_{\text{Refund}_3}$ to the smart contract. 
The smart contract validates the correctness of the proof and allocates the funds designated for Bob's unlocking to him. Then, it settles the remaining fund in a manner consistent with the process outlined in Case 1. 
For the second option, Alice needs to reveal $s_{12}$ and $z_{12}$ (according to the design in Case 2). However, since Chain $\beta$ has completed the unlock operation, $s_{21}$ is also public. Therefore, miners have access to both $s_{12}$ and $s_{21}$, from which they can deduce $s_{11}$ and submit $\mathsf{Tx}_{\text{Appeal}}$ to the smart contract within $T_3$. 
The smart contract verifies whether $s_{11}$ corresponds to the hash lock $h(s_{11})$. If the verification is successful, the available fund (including funds belonging to Alice) is sent to the miner.  
%
%
Based on the above design, the multi-path refund strategy can achieve settlement in different cases. Particularly, Case 3, only allows Alice to assist Bob in completing the refund, thereby resolving the passive offline issue. 

Note that, based solely on $\mathsf{R\text{-}HTLC}$, we still cannot achieve the design goal of the privacy, as adversaries can compromise unlinkability through the funds in transactions. Therefore, we introduce $\mathsf{CCN}$ in the next section to solve the above issue.  

\subsection{R-HTLC Enabled CCN}~\label{subsec:CCN} 
In this subsection, we present the detailed design of $\mathsf{CCN}$, a decentralized cross-chain network that leverages $\mathsf{R\text{-}HTLC}$ to enable secure, privacy-preserving, and multi-hop interactions. The protocol consists of two main components: off-chain interactions and on-chain settlements. 
To illustrate the workflow, we consider a multi-hop transaction from Alice to Emma across ($n$+1) blockchains. 

During the off-chain phase, participants interact through signed receipts. 
Intra-chain receipts are confined to a single channel and support flexible payment amounts. In contrast, cross-chain receipts require consistent values across channels, as practical payment scenarios assume that participants act rationally, meaning that no sender intends to overpay, nor would any receiver accept less than the agreed amount. 
For example, in Chain~$\alpha$, Alice sends 3 uints via intra-chain receipts and 4 uints via cross-chain receipts to Bob. In Chain~$\beta$, Bob may send 10 uints to Carol via intra-chain receipts and must forward 4 uints via cross-chain receipts to ensure consistency with the value received from Alice. 
%
%
Finally, receipts are forwarded to the destination (Emma) and settled on-chain across all participating blockchains. 

%
%
\begin{figure}[t]
	\centering
    \includegraphics[width=0.4\textwidth]{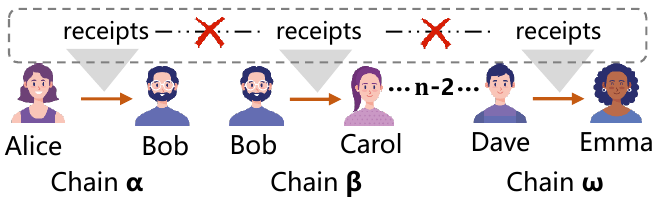}
	\caption{An example of a multi-hop cross-chain transfer. The orange arrows represent the off-chain channels  and the on-chain receipts are unlinkable. 
    } 
    \label{Fig:multi-hop}
\end{figure}

After completing the off-chain interaction, all participants proceed to execute $\mathsf{R\text{-}HTLC}$ to finalize the on-chain settlement. To support multi-hop applications, we provide a general representation of the circuit for generating hash locks, as illustrated in Fig.~\ref{Fig:prepare-circuit}. This design adheres to two key constraints: each hash lock must remain independent, and the unlocking on the $i$\textsuperscript{th} blockchain ($i \in [1, n]$) must be conditionally dependent on the successful unlocking on the $(i+1)$\textsuperscript{th} blockchain.
%
We summarize the logic of this general circuit as follows: 
\begin{framed} \footnotesize
 \vspace{-0.2cm}
    \raggedright
    $\underline{{\bf {The~Logic~of~the~General~Circuit:}}}$ \\
    \hspace{0.5cm}$i = n$; \\
    \hspace{0.5cm} \bf{While}~{$i > 0$}: \\
    \hspace{0.8cm} \bf{If}~{$i = 1$}: \\
        \hspace{1.1cm} $s_{11} = s_{21} \oplus z_{12}$; \\
        \hspace{1.1cm} $h(s_{11}) = SHA\text{-}256(s_{11})$; \\
        \hspace{0.8cm} \bf{If}~{$i \geq 1$}: \\
        \hspace{1.1cm} $s_{i2} = s_{i1} \land s_{(i+1)1}$; \\
       \hspace{1.1cm} $h(s_{i2}, z_{i2}) = SHA\text{-}256(s_{i2}, z_{i2})$; \\
        \hspace{1.1cm} $h(s_{(i+1)1}) = SHA\text{-}256(s_{(i+1)1})$; \\
    \hspace{0.8cm} \bf{If}~{$i > 1$}: \\
        \hspace{1.1cm} $s_{(i+1)1} = s_{i1} \oplus z_{i2}$; \\
        \hspace{0.8cm} $i = i - 1$; 
    \vspace{-0.2cm}
\end{framed}
Following the protocol of $\mathsf{R\text{-}HTLC}$, Emma constructs this circuit based on the above logic and generates the corresponding ZKP for distribution to all involved participants. While the circuit logic is identical for each participant, the public inputs vary depending on their role within the cross-chain channel.  
These public inputs fall into three categories: hash locks, parameters required for completing the unlocking operation, and parameters for handling node offline scenarios. 
For example, Bob, as an intermediary node connected to both Chain~$\alpha$ and Chain~$\beta$, receives the hash locks for both chains as well as the parameters $z_{12}$, $s_{22}$, and $z_{22}$. Here, $z_{12}$ allows Bob to recover $s_{11}$ for completing the unlock, while $s_{22}$ and $z_{22}$ enable him to assist Carol in Chain~$\beta$ in the event of an offline issue. 
The initiators of the cross-chain interaction, such as Alice and Emma, are allowed to possess all relevant information, and we do not consider the possibility of them voluntarily disclosing their private data, as it is generally unrealistic to assume users would willingly reveal sensitive information in practical settings.

During $\mathsf{R\text{-}HTLC}.\mathsf{Lock}$, the locking process proceeds sequentially from Chain~$\alpha$ to Chain~$\omega$, where each locking action initiates an hourglass mechanism. The unlock and refund phases follow the same procedure as introduced in Sec.~\ref{sec:R-HTLC}. Specifically, in $\mathsf{R\text{-}HTLC.Unlock}$, both parties in the same chain exchange signatures to confirm the unlock result. For example, Alice and Bob in Chain~$\alpha$ exchange signatures before proceeding. Unlocking then proceeds in reverse order, from the $(i+1)$\textsuperscript{th} chain back to the $i$\textsuperscript{th}, where each chain uses the received preimage to derive its own. 
If a party in any channel goes offline, the multi-path refund strategy prevents the remaining party in that channel from unilaterally completing the settlement, thus maintaining fairness. 

Note that the complete design of the settlement protocol $\mathsf{R\text{-}HTLC}$,
together with all transaction specifications,
is presented in Sec.~\ref{sec:R-HTLC}.
Fig.~\ref{fig:CCN-Protocol-2} summarizes the overall workflow of $\mathsf{CCN}$. 
\begin{figure*}[t] 
\small
\fbox{%
\parbox{\dimexpr\textwidth-2\fboxsep\relax}{%
\begin{multicols}{2}
\begin{basedescript}{\desclabelwidth{10pt}}

\item[{Off-Chain Interactions:}] 
 /* All nodes interact off-chain in the channel */\\
\quad Suppose two nodes (e.g., Alice and Emma) initiate a cross-chain interaction across $n+1$ chains (see Fig.~\ref{Fig:multi-hop});\\  
\quad All nodes exchange receipts within the channel and upload the final result to the blockchain;\\  
\quad {Initiate settlement based on $\mathsf{R\text{-}HTLC}$.}

\item[R-HTLC.Prepare:]
 /* Generate proofs to create independent hash locks for all nodes */\\
\quad Emma generates the general circuit introduced in Sec.~\ref{subsec:CCN};\\
\quad Emma generates a proof via $\Pi.\mathsf{Prove}$ for each participant with respective public inputs;\\
\quad Each node verifies the public inputs using $\Pi.\mathsf{Verfiy}$. 

\item[R-HTLC.Lock:]
 /* All nodes lock the funds in each channel */\\
\quad On the first chain, Alice sends funds with hash locks $(h(s_{11}),h(s_{12}))$ to smart contract $\xi^\alpha$ via $\mathsf{Tx}_{\text{Lock}}$;\\  
\quad $\xi^\alpha$ locks the funds using $(h(s_{11}), h(s_{12}))$, sets a time lock $T_1$, and updates state to $\mathsf{LOCK}$;\\
\quad $\xi^\alpha$ activates the hourglass mechanism;\\ 
\quad Bob verifies the correctness of the hash locks;\\ 
\quad Repeat sequentially from the first chain to the $(n+1)$\textsuperscript{th} chain;\\ 
\quad \textbf{Waiting for unlock (Hourglass Mechanism):} 
\begin{adjustwidth}{1em}{0pt}
Each node (e.g., Alice) can choose to send $\mathsf{Tx}_{\text{Check}}$ to transfer a portion (e.g., 5) of available funds to another node (e.g., Dave); \\ 
 $\xi^\alpha$ checks if 5 is within the available fund, sends the commitment $\mathsf{cmt}$ to Alice, and updates the available fund;\\
 Alice forwards $\mathsf{cmt}$ to Dave. 
\end{adjustwidth}

\item[R-HTLC.Unlock:]
 /* Nodes in the same channel reach consensus and complete unlocking based on proof */\\
\quad Bob sets his desired funds to 27 on Chain~$\alpha$, computes $h(h(s_{11}), 27)$, and signs it (marked as $\sigma^{\alpha}_{BA}$);\\ 
\quad Bob sends $(h(h(s_{11}), 27), h(s_{11}), 27)$ to Alice;\\
\quad Alice validates $\sigma^{\alpha}_{BA}$, signs $h(h(s_{11}), 27)$ (marked as $\sigma^{\alpha}_{AB}$), and returns it to Bob;\\
\quad Repeat the above process until Dave and Emma in the $(n+1)$\textsuperscript{th} chain reach consensus;\\
\quad Emma generates $\pi^{\mathsf{Sha.}}_{CP}=\Pi.\mathsf{Prove}$;\\                
\quad Emma constructs $\Lambda^{\mathsf{Unl.}}$ (see Fig.~\ref{Fig:unlock-circuit}) and generates $\pi^{\mathsf{Unl.}}$;\\
\quad Emma sends $\pi^{\mathsf{Unl.}}$ via $\mathsf{Tx}_{\text{Unlock}}$ to $\xi^\omega$;\\
\quad $\xi^\omega$ asserts state = $\mathsf{LOCK}$ and no timeout;\\
\quad $\xi^\beta$ verifies $\pi^{\mathsf{Unl.}}$ via $\Pi.\mathsf{Verfiy}$, transfers the locked funds to Emma, and sets state to $\mathsf{UNLOCK}$;\\
\quad Repeat sequentially from the $(n+1)$\textsuperscript{th} chain to the first chain.

\item[R-HTLC.Refund:]
 /* Consider three cases: node offline scenarios and settlement coordination between two adjacent chains (e.g., Chain~$\alpha$ and Chain~$\beta$) */\\
\quad \textbf{Case 1 [Chain~$\alpha$]:}\\
\quad \quad Alice sends $\mathsf{Tx}_{\text{Refund}_1}$;\\
\quad \quad $\xi^\alpha$ asserts state = $\mathsf{UNLOCK}$ and $T > T_1$, then transfers funds per $\mathsf{Tx}_{\text{Check}}$ and remaining funds to Alice, setting state to $\mathsf{REFUND}$;\\
\quad \textbf{Case 2 [Chain~$\alpha$]:}\\
\quad \quad Alice sends $(s_{12},z_{12})$ via $\mathsf{Tx}_{\text{Refund}_2}$;\\
\quad \quad $\xi^\alpha$ asserts state = $\mathsf{LOCK}$ and $T > T_1$, verifies $(s_{12},z_{12})$ against $h(s_{12},z_{12})$, and sets a time lock $T_3$; \\
\quad \quad Miners collect preimages and send $s_{11}$ via $\mathsf{Tx}_{\text{Appeal}}$;\\
\quad \quad $\xi^\alpha$ asserts $T < T_3$, verifies $s_{11}$, transfers the available fund to the miner, and sets state to $\mathsf{APPEAL}$;\\
\quad \quad When $T > T_3$ and state != $\mathsf{APPEAL}$, $\xi^\alpha$ transfers funds per $\mathsf{Tx}_{\text{Check}}$ and remaining funds to Alice, then sets state to $\mathsf{REFUND}$;\\
\quad \textbf{Case 3 [Chain~$\alpha$]:}\\
\quad \quad  Alice constructs $\Lambda^{\mathsf{Unl.}}_{A\xi}$ (cf. Fig.~\ref{Fig:unlock-circuit}), generates $\pi^{\mathsf{Unl.}}_{A\xi}$ via $\Pi.\mathsf{Prove}$, and sends it via $\mathsf{Tx}_{\text{Refund}_3}$ to $\xi^\alpha$;\\
\quad \quad $\xi^\alpha$ asserts state = $\mathsf{LOCK}$ and $T > T_1$, verifies $\pi^{\mathsf{Unl.}}_{A\xi}$, transfers funds to Bob and per $\mathsf{Tx}_{\text{Check}}$ to the node, then sets state to $\mathsf{REFUND}$;\\[1ex]
\quad \textbf{Case 1/2/3 [Chain~$\beta$]:}\\
\quad \quad Bob sends $\mathsf{Tx}_{\text{Refund}_1}$;\\
\quad \quad $\xi^\beta$ asserts state = $\mathsf{UNLOCK}$ and $T > T_2$, transfers funds per $\mathsf{Tx}_{\text{Check}}$ and remaining funds to Bob; \\ 
\quad \quad $\xi^\beta$ sets state to $\mathsf{REFUND}$. 
\end{basedescript}
\end{multicols}\vspace{-10pt}
}%
}
\caption{The workflow of $\mathsf{CCN}$ across multiple blockchains.}  \label{fig:CCN-Protocol-2}
\end{figure*}

\section{Security Analysis}~\label{sec:Analysis}
%
In this section, we analyze the atomicity of $\mathsf{CCN}$ for cross-chain execution, which is a fundamental property required by cross-chain protocols~\cite{tsabary2021mad,thyagarajan2022universal}. 
In addition, we analyze the privacy guarantees of the scheme, while the security against offline attacks is further validated through experimental evaluation.

\subsection{Atomicity} \label{sec:atmoic} 
\begin{mytheorem}[Atomicity]
\label{thm:atomicity}
Assume that the blockchains execute all submitted transactions correctly, the hash functions used in $\mathsf{R\text{-}HTLC}$ are collision-resistant,
and  the employed zero-knowledge proof system is sound.
Then $\mathsf{CCN}$ satisfies cross-chain atomicity. 
\end{mytheorem}
\begin{proof}[Proof Sketch]
We show that if an adversary could violate atomicity in $\mathsf{CCN}$, 
then either some smart contract does not follow the protocol specification, 
or the collision resistance of the hash function is broken,
or the soundness of the zero-knowledge proof system is violated. 

A run of $\mathsf{CCN}$ consists of two layers:
the off-chain channel interactions that produce intra-chain and cross-chain receipts,
and the on-chain settlement that is governed by $\mathsf{R\text{-}HTLC}$ contracts. 
The hourglass mechanism and the commitments $\mathsf{cmt}$ only determine
how much of the locked value is treated as frozen or available, 
but every transfer of funds is still executed through the phases
$\mathsf{R\text{-}HTLC.Lock}$, $\mathsf{R\text{-}HTLC.Unlock}$ or 
$\mathsf{R\text{-}HTLC.Refund}$. 
Therefore, the atomicity of $\mathsf{CCN}$ follows directly
from the atomicity properties enforced by $\mathsf{R\text{-}HTLC}$. 

\noindent{\bf Structure of R-HTLC.} 
For each hop, we have secrets and hash locks. 
The prepare phase enforces the relations
$s_{11} = s_{21} \oplus z_{12}$ and $s_{12} = s_{11} \land s_{21}$ (and their generalizations), 
and binds the hash locks to these secrets through zk-SNARK proofs.
By soundness, any on-chain verification of a proof for these circuits implies that the committed hash values and secrets satisfy the prescribed relations, except with negligible probability. 
By collision resistance, the hash image of any secret is unique for a given hash function. 

Funds can leave a contract in only two ways:
either through an unlock transaction $\mathsf{Tx}_{\text{Unlock}}$
based on a valid proof for the unlock circuits, 
or through one of the refund transactions
$\mathsf{Tx}_{\text{Refund}_1}$, $\mathsf{Tx}_{\text{Refund}_2}$ or
$\mathsf{Tx}_{\text{Refund}_3}$, possibly followed by an appeal transaction
$\mathsf{Tx}_{\text{Appeal}}$.
Since the underlying blockchains are assumed to execute all submitted transactions correctly, no other transfer of locked funds is possible. 

\noindent{\bf Normal Unlock Executions.}
Consider first the case in which all relevant parties remain online and
the unlock phase completes on all involved chains.
On each chain, the amount to be unlocked is agreed off-chain and fixed
through signatures on values of the form $h(h(s_{j1}), v_j)$.
The unlock circuits verify both the correctness of the secret
and the correctness of the signatures.
By the soundness of the proof system and the unforgeability of signatures,
any successful unlock on a chain implies that the underlying secrets
and signatures are consistent with the agreed amounts. 
Since the secrets across chains are linked through the prepare constraints
(for example $s_{11} = s_{21} \oplus z_{12}$),
revealing a secret that allows an adversary to unlock value on one chain
also enables the honest counterparty to unlock the corresponding value on the other chain.
Hence, in a normal execution, either all chains unlock according to the agreed balances,
or no chain unlocks, which satisfies atomicity.

\noindent{\bf Refund Cases and Hourglass Mechanism.}
We consider executions in which some parties go offline, 
and settlement proceeds via the refund mechanisms. 
We follow the three cases described in $\mathsf{R\text{-}HTLC.Refund}$. 

In Case 1, both sides have completed the unlock operation. 
The refund transaction $\mathsf{Tx}_{\text{Refund}_1}$ then returns
the remaining available funds, after accounting for commitments and unlocked value, 
to the original locker. 
Since all transfers are determined by previously recorded 
$\mathsf{Tx}_{\text{Lock}}$ and $\mathsf{Tx}_{\text{Check}}$ transactions,
and the contract only distributes the remaining balance, 
no honest party can lose funds while another gains more than agreed. 

In Case 2, neither side completes the unlock and the locker reveals $(s_{12},z_{12})$ 
through $\mathsf{Tx}_{\text{Refund}_2}$.
The contract checks that these values match the hash lock $h(s_{12},z_{12})$. 
If they do, the contract proceeds to refund the funds after $T_3$.
If an adversary attempts to cheat by providing inconsistent values,
the proof or hash check fails and no funds are released. 
Thus, either the honest locker eventually recovers the locked funds,
or no transfer occurs.

In Case 3, one side (for example Carol) completes the unlock on one chain,
while the other side (Bob) fails to unlock before timeout. 
The protocol offers two paths. 
Through $\mathsf{Tx}_{\text{Refund}_3}$, Alice can generate a proof
for a circuit similar to $\Lambda^{\mathsf{Unl.}}_{B\xi}$ and help Bob
obtain his rightful funds.
If Alice instead behaves maliciously and uses $\mathsf{Tx}_{\text{Refund}_2}$ to claim the funds intended for Bob, she must reveal $(s_{12},z_{12})$.
At that point the corresponding secret $s_{21}$ has already been revealed
on the other chain, so miners can derive $s_{11}$ and submit
$\mathsf{Tx}_{\text{Appeal}}$.
By soundness and collision resistance, the contract will accept only a correct $s_{11}$ before $T_3$, in which case the remaining fund (including Alice's share) is transferred to the miner. 
Therefore, Alice cannot both obtain the remote funds and recover her local funds, and thus gains no benefit from deviating from the protocol. 
%

The hourglass mechanism allows a participant to gradually obtain
available funds before the timeout expires.
However, these funds do not correspond to immediate on-chain transfers.
Instead, every expenditure of available funds is recorded as a commitment
$\mathsf{cmt}$ issued by the smart contract, which specifies the receiver
and the committed amount. 
Such commitments only become effective during the subsequent refund phase,
where they are validated and settled by the contract. 
Consequently, although a participant may temporarily utilize part of its locked funds 
through commitments, the contract maintains an explicit record of all outstanding 
commitments and enforces that the sum of frozen funds, available funds,
and committed amounts never exceeds the originally locked value.
Any attempt to overspend or reuse the same funds in multiple commitments
is rejected by the contract. 
Therefore, the hourglass mechanism does not enable immediate double spending
or inconsistent transfers, and all interim fund usage is ultimately
accounted for and resolved during the refund process. 


\noindent{\bf Reduction to Cryptographic Assumptions.}
Suppose that there exists a PPT adversary
that causes a violation of cross-chain atomicity in $\mathsf{CCN}$ with
non-negligible probability. Then, by the above arguments, there must exist an execution in which a smart contract releases funds in a way that is inconsistent with the specified protocol logic, or an execution in which an adversary unlocks or refunds funds while preventing an honest counterparty from obtaining the corresponding value. 
In the first case, the smart contract does not behave according to its specification, contradicting our assumption. 
In the second case, the adversary must either: produce two different secrets that map to the same hash value, which breaks collision resistance, or convince a contract to accept an invalid proof for one of the prepare or unlock circuits, which violates the soundness of the zero-knowledge proof system. 
Hence, any non-negligible violation of atomicity would imply a break of one of the stated assumptions. 

Therefore, under the given assumptions, $\mathsf{CCN}$ satisfies cross-chain atomicity.
\end{proof}

\subsection{Unlinkability} 
\label{sec:unlink-main} 

\begin{mytheorem}[Unlinkability] \label{thm:unlinkability} 
Consider the security game $\mathsf{Game}^{\mathsf{Unlink}}_{\mathcal{A}}(\lambda)$ in Definition~\ref{Def:unlinkability}.
Assuming that the hash function is collision-resistant and the zero-knowledge proof protocol is secure,  $\mathsf{CCN}$ achieves \emph{unlinkability}, that is, 
the adversary’s distinguishing advantage $\mathsf{Adv}^{\mathsf{Unlink}}_{\mathcal{A}}(\lambda)$ is negligible in $\lambda$.
Equivalently, no probabilistic polynomial-time adversary can determine which sender--receiver pair was involved in the challenge transaction
with probability non-negligibly greater than $1/2$.
\end{mytheorem} 

\begin{proof}[Proof Sketch]
Let $\mathcal{A}$ be a PPT adversary participating in 
$\mathsf{Game}^{\mathsf{Unlink}}_{\mathcal{A}}(\lambda)$.
As in Definition~\ref{Def:unlinkability}, the challenger initializes the system,
selects two distinct sender–receiver pairs $(S_0,R_0)$ and $(S_1,R_1)$,
samples a random bit $b \in \{0,1\}$ and simulates a cross-chain payment from $S_b$ to $R_b$,
producing the observable transcript $\mathcal{O}$, which includes
all on-chain transactions and publicly visible protocol messages.
The adversary outputs a guess $b'$.
We argue that the advantage
$\mathsf{Adv}^{\mathsf{Unlink}}_{\mathcal{A}}(\lambda) 
= \bigl|\Pr[b' = b] - \tfrac{1}{2}\bigr|$ is negligible.

\noindent{\bf Address-based Linkage.} 
In $\mathsf{CCN}$, payments are routed through intermediary participants in a hop-by-hop manner. As a result, the direct link between the external addresses of $S_b$ and $R_b$ 
does not appear as a single on-chain transaction. Unless all participants on a single routing path collude and share their private information, 
which is excluded by the standard threat model assumed in existing privacy-preserving
cross-chain payment schemes,  
the adversary does not observe a complete end-to-end mapping between $S_b$ and $R_b$.
Hence, address-level analysis alone does not provide a non-negligible distinguishing advantage.  

\noindent{\bf Hash Locks and ZKPs.}
Each hop  uses independent secrets and publishes
only their hash images. 
Cross-hop consistency is enforced by ZKPs for the general circuit 
and the unlock circuit $\Lambda^{\mathsf{Unl.}}$.
By the zero-knowledge property, the proofs can be replaced by simulated proofs 
without changing the distribution of $\mathcal{A}$'s view in any computationally noticeable way. 
Similarly, since the adversary does not learn the preimages underlying these hash values 
(except in the appeal cases), 
the hash outputs can be replaced by random but internally consistent strings 
without giving $\mathcal{A}$ a non-negligible distinguishing advantage, 
unless the hash function or the zero-knowledge proof protocol is insecure. 

\noindent\textbf{Amounts.}
In $\mathsf{CCN}$, the value settled on-chain is computed as the aggregate of multiple off-chain receipts. Specifically, each settlement amount consists of both intra-chain receipts
and cross-chain receipts accumulated along the execution of the protocol. While the cross-chain receipts are determined by the protocol logic,
the intra-chain receipts can be generated arbitrarily and are inherently unpredictable.
As a result, the final aggregated amount committed on-chain appears as a random-looking value from the adversary’s perspective.

Consequently, even when observing all on-chain transaction values, the adversary cannot reliably correlate two executions or infer whether they originate from the same sender–receiver pair.
Thus, amount information in $\mathcal{O}$ does not provide any non-negligible advantage for breaking unlinkability. 

\noindent{\bf Reduction Argument.}
Suppose, for the sake of contradiction, that there exists a PPT adversary $\mathcal{A}$
such that $\mathsf{Adv}^{\mathsf{Unlink}}_{\mathcal{A}}(\lambda)$ is non-negligible.
Then one can build a distinguisher that uses $\mathcal{A}$ as a subroutine to distinguish
between real and simulated proofs, or between real hash outputs and random values,
depending on which part of the transcript $\mathcal{A}$ exploits.
This would contradict either the security of the underlying hash function
or the zero-knowledge guarantees of the proof system. 
Therefore, any distinguishing advantage must be negligible, and $\mathsf{CCN}$
achieves unlinkability in the sense of Definition~\ref{Def:unlinkability}. 
\end{proof}

\section{Implementation and Performance Evaluation}~\label{sec:Test} 
%
%
In this section, we present the implementation of $\mathsf{CCN}$ and test its performance. These experiments complement the security analysis and validate the practicality of handling offline scenarios. 



\subsection{Implementation}\label{subsec: Implementation}
%
We deploy $\mathsf{CCN}$ in a multi-chain environment, and our implementation mainly includes the following components: 

\noindent{\bf Multi-Chain Networks.} 
We implement our scheme on both Ethereum and Cosmos to evaluate its practical feasibility. All entities described in Sec.~\ref{subsec: model} are instantiated as blockchain nodes. For Ethereum, we use the go-ethereum client\cite{codeEthereum}, where each node is simulated in an isolated Docker container. Smart contracts are written in Solidity~\cite{codeSolidity} and deployed on Ethereum to enforce conditional transactions. The implementation includes over 200 lines of Solidity code covering all contract functionalities. Each node interacts with the contracts using web3.py, enabling deployment and invocation of contract functions. For Cosmos, we adopt Juno~\cite{juno}, a blockchain built on the Cosmos SDK and secured by the Tendermint consensus protocol.
To streamline deployment and testing, we implement over 150 lines of Bash scripts to automate blockchain initialization and transaction execution. 
Smart contracts on Juno are developed in Rust and compiled into WebAssembly (Wasm) using Cargo, with over 400 lines of code realizing the core logic. 
%
%
%

We deploy the blockchain using multiple cloud servers from five countries (China, Germany, Japan, Singapore, and the United States). 
The servers in China are configured with S6.2XLARGE16 specifications, each operating on the Ubuntu 20.04 system with an Intel Ice Lake processor. They are equipped with 8 vCPUs running at a frequency of 3.3 GHz and 16 GB of RAM. 
The other servers feature S5.2XLARGE16 specifications, run on the Ubuntu 20.04 system with an Intel Xeon Cascade Lake 8255C/Intel Xeon Cooper Lake processor, and have 8 vCPUs operating at a frequency of 3.1 GHz, along with 16 GB of RAM.  
We initiate 4 Docker instances on each server to form multiple blockchains. 
%


%
%

\noindent{\bf zk-SNARK.} We utilize xjsnark~\cite{codeXjsnark} and libsnark~\cite{codeLibsnark} to implement zk-SNARK. 
xJsnark is a high-level framework for users to write the circuits in zkSNARK, and  
we use it to obtain the circuits in our scheme, which are implemented with 2,000+ lines of code. 
Then, we adopt the Groth16 algorithm~\cite{groth2016size} in libsnark to realize the initialization ($\mathsf{Setup}$), generation ($\mathsf{Prove}$), and verification ($\mathsf{Verify}$) of zk-SNARK. A prover first generates a circuit based on xJsnark, then utilizes libsnark to initialize and generate a proof. 
Subsequently, a verifier can call smart contracts or use libsnark to validate the generated proof. 
%
Further specifics regarding the experimental setup will be provided before presenting the performance evaluation results. 

\subsection{Performance Evaluation}\label{subsec: Performance}
In this subsection, we evaluate the performance of $\mathsf{CCN}$ in a multi-chain environment. Specifically, we focus on the following aspects: (i) the advantages of $\mathsf{CCN}$ in mitigating offline issues; (ii) the additional resource overhead introduced by the security and privacy mechanisms, assessed in terms of latency, throughput, and gas consumption; and (iii) the feasibility of deploying $\mathsf{CCN}$ across heterogeneous blockchains. 

\begin{figure}[ht]
    \centering
    \subfigure[{\bf Active offline issue}]{\includegraphics[width=0.23\textwidth]{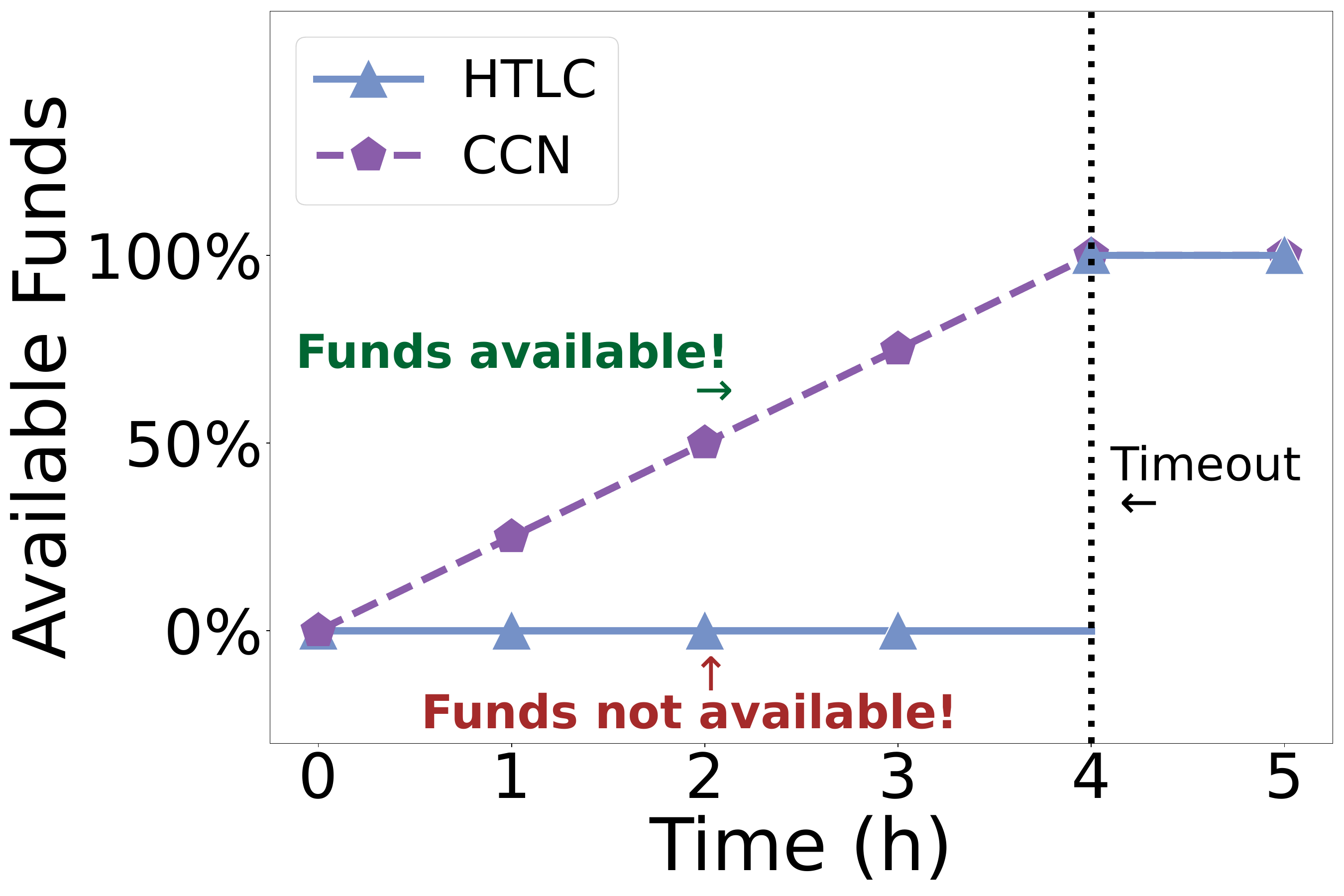}}
    \subfigure[{\bf Passive offline issue}]{\includegraphics[width=0.23\textwidth]{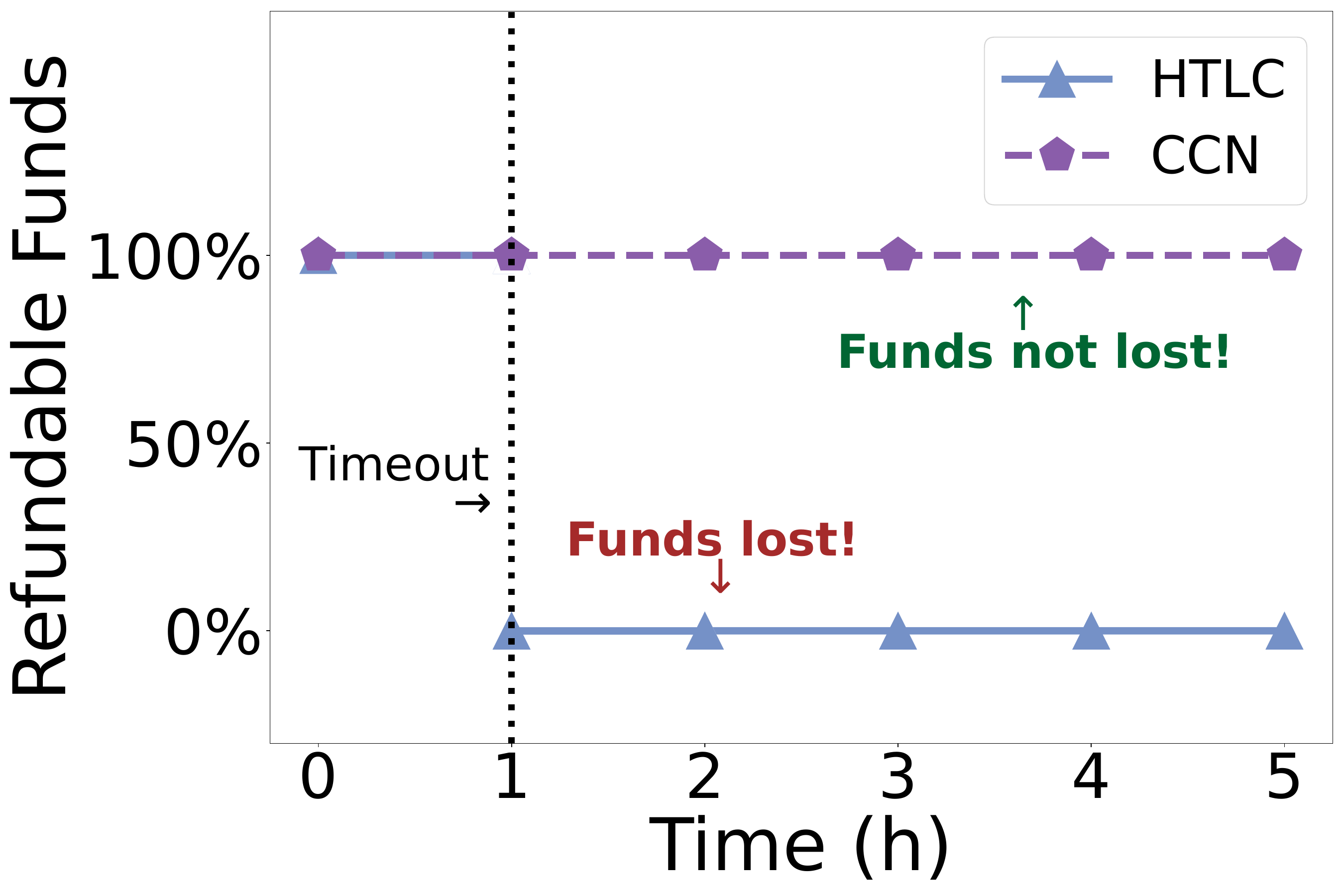}}
    \caption{Impact of active and passive offline issues on HTLC-based protocols and the mitigation by $\mathsf{CCN}$. } 
    \label{Fig:experiments-offline}
\end{figure}
\noindent{\bf Mitigating Offline Issues.}
As demonstrated in Sec.~\ref{subsec: model}, we simulate the offline issues encountered by HTLC-based schemes~\cite{tsabary2021mad,guo2023cross} and analyze their consequences. In this part, we evaluate the capability of $\mathsf{CCN}$ in mitigating these issues. 
Specifically, Fig.~\ref{Fig:experiments-offline} (a) illustrates the scenario of an active offline issue, in which a malicious node deliberately becomes unavailable. The vertical axis denotes the available funds of the honest participant. It can be observed that during the time window preceding the  timeout, the honest party is unable to access its locked assets. In contrast, $\mathsf{CCN}$ enables continuous availability of the honest party's funds, thereby mitigating the issue. 
Fig.~\ref{Fig:experiments-offline} (b) presents the passive offline issue, where an honest node becomes temporarily unavailable due to unexpected network or device failures. After the timeout, the refundable funds become vulnerable to being claimed by the counterparty. However, $\mathsf{CCN}$ ensures these funds remain secure, preventing financial loss even under uncertain recovery periods. 
\begin{figure}[ht]
    \centering
    \subfigure[{\bf Latency}]{\includegraphics[width=0.23\textwidth]{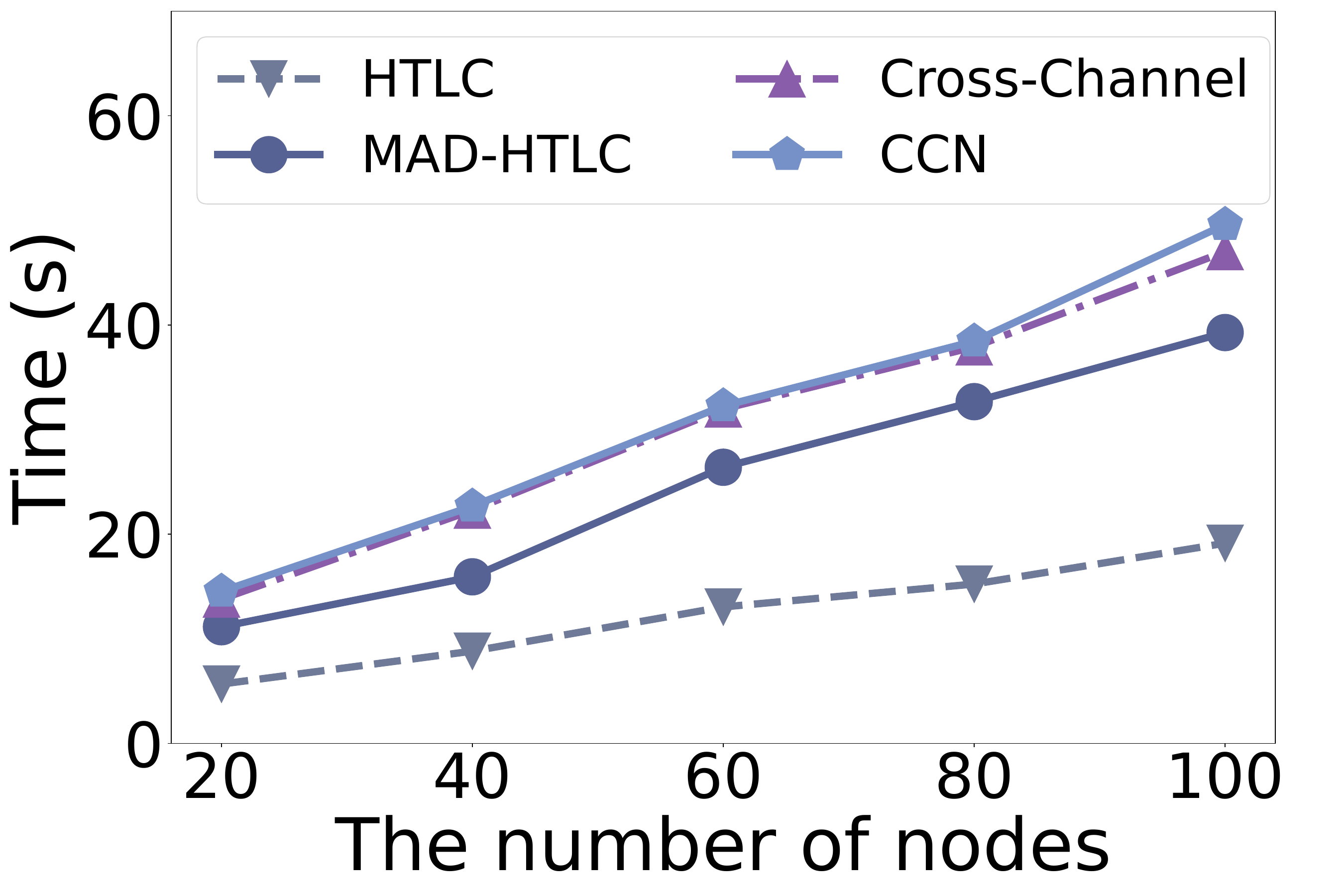}}
    \subfigure[{\bf Throughput}]{\includegraphics[width=0.23\textwidth]{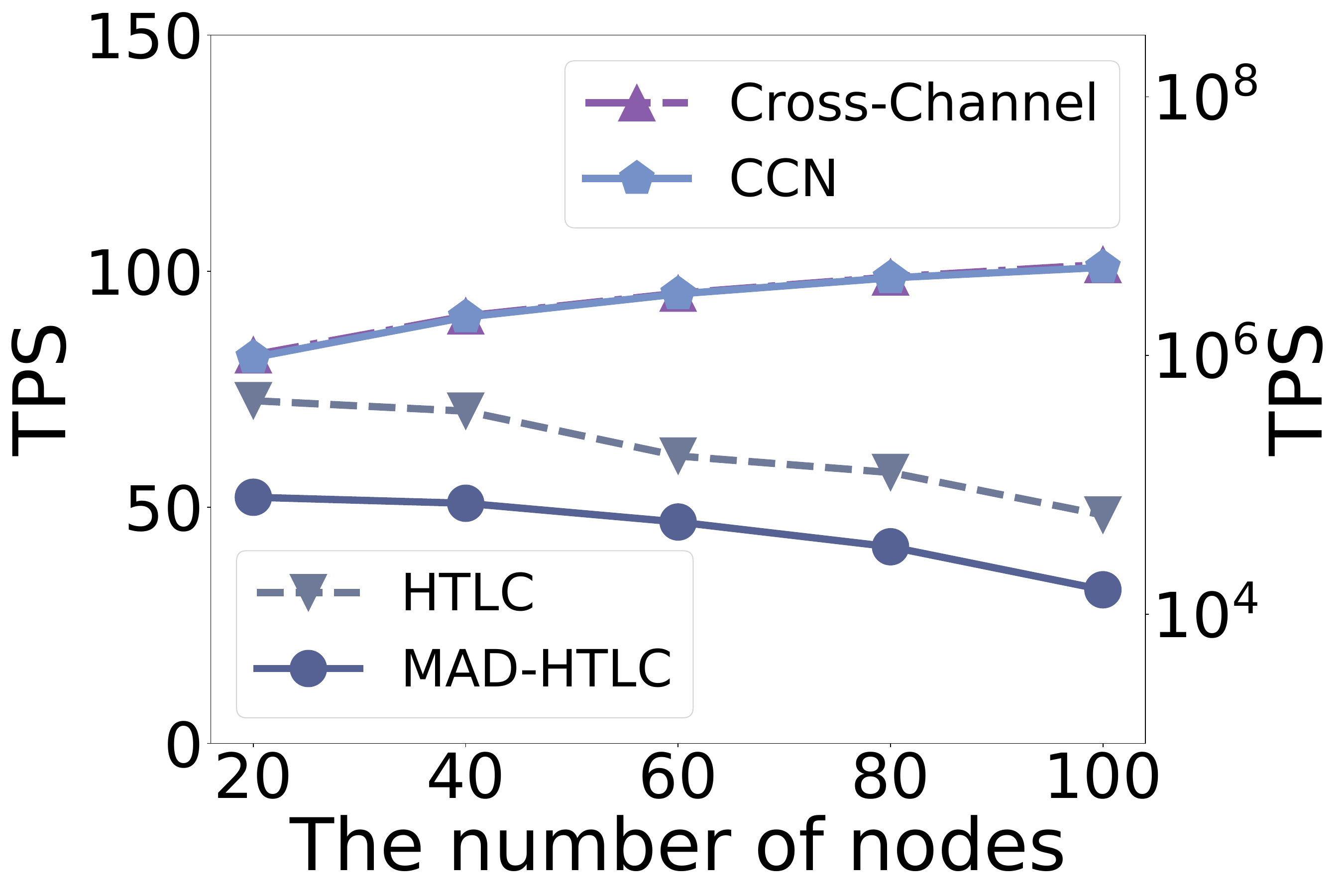}}
    \caption{Comparison experiments on latency (left) and throughput (right). } 
    \label{Fig:Compare}
\end{figure}
\noindent{\bf Comparative Analysis of On-Chain Resource Overhead.} We make a comparative analysis, focusing on on-chain latency, throughput, and gas consumption, against the three most relevant decentralized cross-chain solutions:  HTLC, MAD-HTLC~\cite{tsabary2021mad}, and Cross-Channel~\cite{guo2023cross}. 
We build two blockchains based on go-ethereum, configuring each with 20 to 100 nodes, while setting the number of channels to 10 times the number of nodes. This means that as the node count increases, the number of channels also increases. Each channel remains open for an average of 100 seconds and can transmit $3 \times 10^6$ $Tr$/s. 
%

As illustrated in Fig.~\ref{Fig:Compare}, the delay and throughput results demonstrate that, with the linear increase in the number of nodes, both $\mathsf{CCN}$ and Cross-Channel exhibit a proportional increase in TPS, highlighting their scalability. 
Then, we test the gas consumption required for the successful execution of different solutions. As shown in Table~\ref{Table:compareCross}, $\mathsf{CCN}$ consumes about 3,200,000 gas to process $N$ cross-chain interactions and HTLC (MAD-HTLC) needs to take about 430,000$\times N$ (750,000$\times N$) gas to process the same volume of operations. 
%

%
\begin{table}[ht]  \footnotesize
\centering
\caption{Comparison experiments on gas consumption.} 
\label{Table:compareCross} 
\begin{threeparttable}  
\begin{tabularx}{0.98\linewidth}{|>{\centering\arraybackslash}m{0.3cm}|>{\centering\arraybackslash}m{1.55cm}|>{\centering\arraybackslash}m{1.54cm}|>{\centering\arraybackslash}m{1.7cm}|>
{\centering\arraybackslash}m{1.4cm}|}
\hline     & HTLC &  MAD-HTLC & Cross-Channel & CCN \\   
\hline  Gas  & $429,532 \times N$ & $ 758,095 \times N$ & $1,330,858$ & $3,209,972$ \\  
\hline   
\end{tabularx}
\end{threeparttable}  
\end{table}
\noindent{\bf Experiments on Heterogeneous Blockchains.}
To further verify the adaptability of our scheme across heterogeneous blockchain platforms, we extend our implementation from Ethereum to Juno, a Cosmos-based chain featuring a distinct technical stack. While Ethereum adopts the EVM, Juno utilizes the CosmWasm virtual machine and operates on a Tendermint consensus, offering a fundamentally different execution and consensus environment. 
%
To assess the practicality of our deployment, we measure the gas consumption of core operations on Juno: uploading the contract (wasm store) consumes approximately 1,617,000 gas, instantiating the contract (wasm instantiate) requires around 152,000 gas, and successful protocol execution incurs an average gas cost of 507,000 per node.

\section{Conclusion}~\label{sec:Conclusion}
In this paper, we propose $\mathsf{CCN}$, a cross-chain channel network designed to enhance security and privacy in multi-hop cross-chain interactions. $\mathsf{CCN}$ effectively addresses both offline availability issues and unlinkability challenges that arise in cross-chain protocols, ranging from single-hop to multi-hop settings. 
We first introduce the architecture and design principles of the $\mathsf{CCN}$ framework. We then present the construction of $\mathsf{R\text{-}HTLC}$ and its integration into $\mathsf{CCN}$, enabling resilience against both active and passive offline behaviors while preserving transaction unlinkability. Finally, we conduct extensive simulations to evaluate the practicality and performance of our approach. 
Future research will focus on exploring how to realize multi-hop cross-chain interactions
without relying on smart contracts,
and on optimizing ZKPs to reduce the associated resource consumption.


\bibliographystyle{ieeetr}
 
\bibliography{IEEEabrv, reference}

\end{document}